\documentclass[]{interact}

\usepackage{epstopdf}% To incorporate .eps illustrations using PDFLaTeX, etc.
\usepackage[caption=false]{subfig}% Support for small, `sub' figures and tables

\usepackage[numbers,sort&compress]{natbib}% Citation support using natbib.sty
\usepackage{algorithm}
\usepackage{algorithmic}
\usepackage{multirow}
\usepackage{rotating}
\bibpunct[, ]{[}{]}{,}{n}{,}{,}% Citation support using natbib.sty
\renewcommand\bibfont{\fontsize{10}{12}\selectfont}% Bibliography support using natbib.sty

\theoremstyle{plain}% Theorem-like structures provided by amsthm.sty
\newtheorem{theorem}{Theorem}[section]
\newtheorem{lemma}[theorem]{Lemma}

\theoremstyle{definition}

\theoremstyle{remark}

\begin{document}

\title{Optimal Poisson subsampling for quantile regression with large-scale longitudinal data}

\author{
\name{Chunjing Li\textsuperscript{a}, Jiahui Zhang\textsuperscript{a} and Xiaohui Yuan\textsuperscript{a}$^*$\thanks{CONTACT Xiaohui Yuan. Email: yuanxh@ccut.edu.cn}}
\affil{\textsuperscript{a}School of Mathematics and Statistics, Changchun University of Technology, Changchun, People’s Republic of
China}
}

\maketitle

\begin{abstract}
To alleviate the computational challenges posed by large-scale longitudinal data, we propose an optimal Poisson subsampling method for quantile regression. The asymptotic properties of the subsampling estimator are established, and the optimal Poisson subsampling probabilities are derived under the A-optimality criterion. For practical implementation, a computationally feasible algorithm is developed for parameter estimation. Furthermore, regularized parameter estimation is incorporated into the optimal Poisson subsampling framework to identify important covariates. Numerical simulations demonstrate that the proposed optimal Poisson subsampling method outperforms the uniform Poisson subsampling method, and the regularized estimator performs satisfactorily. An empirical analysis of depression data further illustrates the practical performance of the proposed method.
\end{abstract}


\section{Introduction}

Longitudinal data are widely encountered in biomedicine, economics, the social sciences, and many other fields. Such data combine cross-sectional and time-series characteristics, with repeated observations from the same subject typically exhibiting within-subject correlation. With rapid advances in data collection and storage technologies, the scale of longitudinal datasets has increased substantially. Representative examples include the China Health and Retirement Longitudinal Study (CHARLS) data \cite{Han26} and the Beijing air quality dataset \cite{Gao25,Fan24}. Consequently, the direct application of conventional statistical methods to large-scale longitudinal datasets can be computationally prohibitive or even infeasible. Therefore, the development of scalable methods for analyzing large-scale longitudinal data under limited computational resources has attracted increasing attention.

To address the computational challenges posed by massive datasets, various approaches have been developed, including divide-and-conquer methods \cite{Lin11,Ma22}, online updating methods for streaming data \cite{Schifano16,Luo23}, and subsampling methods \cite{Fithian14,Wang18,Li24,Shan24}. Among these approaches, subsampling methods aim to draw an informative subset from the full dataset and conduct parameter estimation based on the resulting subsample. Considerable progress has been made in the development of subsampling methods for massive data analysis. For example, Fithian and Hastie \cite{Fithian14} introduced Poisson subsampling and applied it to logistic regression. Wang \textit{et al.} \cite{Wang18} developed an optimal subsampling method for logistic regression and proposed a two-step algorithm based on the A-optimality criterion. Optimal subsampling methods were subsequently extended to softmax regression, quantile regression, and generalized linear models; see, e.g., \cite{Yao19,Ai21b,Ai21,Yue24}. More recently, considerable attention has been devoted to response-free subsampling methods, in which the subsampling probabilities depend only on the covariates rather than on the response variables, or at most on a very small fraction of the response variables. Such methods are particularly useful when the responses are unavailable or costly to obtain and also allow the subsampling probabilities to be determined without accessing the full response data; see, e.g., \cite{Shan25,Zhang21,Ma15}. Moreover, regularized estimation within the optimal subsampling framework has been investigated. Specifically, Li \textit{et al.} \cite{Li24b} proposed a Poisson subsampling method based on the least product relative error loss for distributed systems and constructed a distributed regularized LPRE estimator for high-dimensional variable selection.
 
The studies reviewed above have primarily focused on independent data. By comparison, relatively few studies have investigated subsampling methods for longitudinal data, and the existing work has focused mainly on mean regression. Specifically, Wang \textit{et al.} \cite{Wang23} developed an optimal subsampling method whose resulting estimator approximates the full-data maximum likelihood estimator. Han and Fu \cite{Han26} proposed an optimal subsampling method based on marginal models for large-scale longitudinal data. 

However, mean regression methods can be sensitive to outliers. As a more robust alternative, quantile regression provides a flexible framework for longitudinal data analysis by characterizing different quantiles of the conditional distribution of the response variable. Quantile regression for longitudinal data has been extensively studied over the past several decades. For example, Fu and Wang \cite{Fu12} proposed a linear quantile regression method for longitudinal data and employed an induced smoothing method to estimate the regression parameters and covariance matrix. Wang and Shan \cite{Wang21} used copula functions to characterize within-subject dependence in composite quantile regression for longitudinal data. Zu \textit{et al.} \cite{Zu23} developed a quantile-based penalized GEE approach and applied it to blood pressure analysis. Nevertheless, directly applying these methods to large-scale longitudinal datasets can impose a substantial computational burden. To address this issue, this paper develops an optimal Poisson subsampling method for quantile regression with large-scale longitudinal data. Regularized parameter estimation is further conducted based on the resulting optimal subsample to identify important covariates.

The main contributions of this paper are as follows. (i) An optimal Poisson subsampling method for quantile regression with large-scale longitudinal data is proposed, and the theoretical properties of the estimator are established. (ii) A computationally feasible algorithm is developed for practical implementation of the optimal subsampling method. (iii) Regularized parameter estimation is performed to identify important covariates within the optimal subsampling framework.

The remainder of this paper is organized as follows. Section 2 introduces the subsampling method for quantile regression with longitudinal data and establishes the asymptotic properties of the proposed estimator. Section 3 derives the optimal Poisson subsampling probabilities under the A-optimality criterion. Section 4 develops regularized parameter estimation within the optimal Poisson subsampling framework. Section 5 evaluates the effectiveness of the proposed method through simulation studies. Section 6 presents a real-data application to demonstrate the practical usefulness of the proposed method. Section 7 concludes the paper and discusses possible future research directions. All proofs are provided in the Supplementary Material.

\section{Poisson subsampling for quantile regression}

In this section, we first review the quantile regression model for longitudinal data, then introduce weighted smoothed estimating functions and establish the asymptotic properties of the resulting estimator.

\subsection{Model and estimating functions}

For $i=1,\cdots,n$, let $\boldsymbol{Y}_i=(y_{i1},\cdots,y_{im_i})^\top$ and $\boldsymbol{X}_i=(\boldsymbol{x}_{i1},\cdots,\boldsymbol{x}_{im_i})^\top$
denote the response vector and covariate matrix, respectively, where $y_{ij}$ is the response at the $j$th measurement occasion and $\boldsymbol{x}_{ij}\in\mathbb{R}^p$ is the corresponding covariate vector for $j=1,\cdots,m_i$. Assume that $m_i=m$ for all $i=1,\cdots,n$. Consider the following conditional quantile model for the response $y_{ij}$:
\[
Q_\tau(y_{ij} \mid \boldsymbol{x}_{ij}) = \boldsymbol{x}_{ij}^\top \boldsymbol{\beta}_0,
\]
where $\tau\in(0,1)$ is a given quantile level and $\boldsymbol{\beta}_0$ denotes the true regression parameter vector. Throughout this paper, the covariates $\boldsymbol{x}_{ij}$ are treated as non-random. Observations from the same subject are allowed to be correlated, whereas observations from different subjects are assumed to be mutually independent. Let $\varepsilon_{ij}=y_{ij}-\boldsymbol{x}_{ij}^\top\boldsymbol{\beta}_0$ denote an error term with a continuous distribution satisfying $P(\varepsilon_{ij}\leq0)=\tau$ and having an unspecified density function $f_{ij}(\cdot)$. As in Zu \textit{et al.} \cite{Zu23} and Fu \textit{et al.} \cite{Fu15}, we consider the following estimating functions
\begin{equation*}
   \boldsymbol{S}_n(\boldsymbol{\beta}) = \frac{1}{n}\sum_{i=1}^{n} \boldsymbol{X}_i^\top \boldsymbol{\Gamma}_i \boldsymbol{R}^{-1} \big[\tau - I\left(\boldsymbol{Y}_i \leq \boldsymbol{X}_i\boldsymbol{\beta}\right) \big].
\end{equation*}
Here, $\boldsymbol{\Gamma}_i=\operatorname{diag}\bigl(f_{i1}(0),\cdots,f_{im}(0)\bigr)$, and $\tau - I\left(\boldsymbol{Y}_{i} \leq \boldsymbol{X}_{i} \boldsymbol{\beta}\right) = \big( \tau - I\left(y_{i1} \leq \boldsymbol{x}_{i1}^\top \boldsymbol{\beta}\right), \cdots, \tau - I\left(y_{im} \leq \boldsymbol{x}_{im}^\top \boldsymbol{\beta}\right) \big)^\top$. The matrix $\boldsymbol{R}$ is an $m\times m$ working correlation matrix with a prespecified structure. Common choices include first-order autoregressive (AR(1)), exchangeable (EX), and first-order moving average (MA(1)) structures.

To tackle the non-differentiability issue of the quantile check function, the induced smoothing method of Brown and Wang \cite{Brown05} is adopted, and thus the corresponding smoothed estimating functions are given by
\begin{equation}
\label{eq:Sn_smoothed}
\boldsymbol{S}_n^\Phi(\boldsymbol{\beta}) = \frac{1}{n}\sum_{i=1}^{n}\boldsymbol{X}_i^\top \boldsymbol{\Gamma}_i \boldsymbol{R}^{-1} \bigg[ \Phi\left( \frac{\boldsymbol{Y}_i - \boldsymbol{X}_i \boldsymbol{\beta}}{h} \right) - (1 - \tau) \bigg],
\end{equation}
where $\Phi(\cdot)$ denotes the standard normal cumulative distribution function, and $h>0$ is a bandwidth that shrinks to zero as the sample size increases. For simplicity, $h$ is not treated as an additional tuning parameter, consistent with the treatment in Zu \textit{et al.} \cite{Zu23}. The estimator $\hat{\boldsymbol{\beta}}$ based on the full dataset $\mathcal{F}_n=\{(\boldsymbol{X}_i,\boldsymbol{Y}_i)\}_{i=1}^n$ is obtained by solving $\boldsymbol{S}_n^\Phi(\boldsymbol{\beta})=\boldsymbol{0}$. When $n$ is sufficiently large, directly solving the estimating equations can be computationally burdensome. To reduce the computational cost, we introduce a Poisson subsampling method in the next subsection.

\subsection{Weighted smoothed estimating functions}

Following \cite{Ai21b,Shan25,Zhang21}, let $\pi_i = \pi_i(\{\boldsymbol{X}_i\}_{i=1}^n)$ denote the inclusion probability of the $i$-th subject. Such a subsampling strategy is particularly useful when response measurements are difficult or costly to obtain. Under Poisson subsampling, a subsample with expected size $r$ is drawn from the $n$ subjects. Let $r^*$ denote the realized subsample size. Then $\operatorname{E}(r^*)=\sum_{i=1}^{n}\pi_i=r$. By incorporating inverse probability weights into (\ref{eq:Sn_smoothed}), the weighted smoothed estimating functions are defined as
\begin{equation}
\label{eq:Sr_smoothed}
\boldsymbol{S}_r^\Phi(\boldsymbol{\beta}) = \frac{1}{n}\sum_{i=1}^{n} \frac{\delta_i}{\pi_i} \boldsymbol{X}_i^\top \boldsymbol{\Gamma}_i \boldsymbol{R}^{-1} \bigg[ \Phi\left( \frac{\boldsymbol{Y}_i - \boldsymbol{X}_i \boldsymbol{\beta}}{h} \right) - (1 - \tau) \bigg],
\end{equation}
where $h$ is set to $r^{-1/2}$ and $\delta_i\in\{0,1\}$ denotes the subsampling indicator variable. If $\delta_i=1$, the $i$th subject is included in the subsample, whereas $\delta_i=0$ indicates that the corresponding subject is excluded.

Let $\tilde{\boldsymbol{\beta}}$ denote the estimator obtained by solving the weighted smoothed estimating equations $\boldsymbol{S}_r^\Phi(\boldsymbol{\beta})=\boldsymbol{0}$. To establish the asymptotic properties of $\tilde{\boldsymbol{\beta}}$, we impose the following regularity conditions.
\begin{itemize}
    \item[(C1)] The true regression coefficient $\boldsymbol{\beta}_0$ lies in a compact set $\Lambda = \{\boldsymbol{\beta} \in \mathbb{R}^p: \|\boldsymbol{\beta}\| \leq B\}$ for some large constant $B$.

    \item[(C2)] For each $i=1,\cdots,n$ and $j=1,\cdots,m$, the density function $f_{ij}(\cdot)$ is continuous and uniformly bounded, and its derivative $f_{ij}'(\cdot)$ exists and is also uniformly bounded.
        
    \item[(C3)] The eigenvalues of the working correlation matrix $\boldsymbol{R}$ are uniformly bounded away from $0$ and $\infty$.
    
   \item[(C4)] As the sample size $n \to \infty$, the following two convergence results hold: (i) The matrix $n^{-1}\sum_{i=1}^n \boldsymbol{X}_i^\top \boldsymbol{\Gamma}_i \boldsymbol{R}^{-1} \boldsymbol{\Gamma}_i \boldsymbol{X}_i$ converges to a positive definite matrix; (ii) The matrix $n^{-1}\sum_{i=1}^{n}\boldsymbol{X}_i^\top \boldsymbol{\Gamma}_i \boldsymbol{R}^{-1}\boldsymbol{V}_0\boldsymbol{R}^{-1}\boldsymbol{\Gamma}_i\boldsymbol{X}_i$ also converges to a positive definite matrix, where $\boldsymbol{V}_0=\operatorname{E}\{[\tau-I(\boldsymbol{\varepsilon}_i\leq0)][\tau-I(\boldsymbol{\varepsilon}_i\leq0)]^T\}$.

    \item[(C5)] $\max\limits_{1\leq i\leq n} \|\boldsymbol{X}_i\|= O(1)$.
    
    \item[(C6)] $\max\limits_{1\leq i\leq n} \frac{r}{n\pi_i} = O(1)$.
\end{itemize}

Condition (C1) is a regularity condition for establishing consistency, as exemplified by Newey and McFadden \cite{Newey94}. Condition (C2) is also used in the studies of \cite{Fu12}. Condition (C3) provides the required assumption for the working correlation matrix, and this condition is formulated following Zu \textit{et al.} \cite{Zu23}. Condition (C4)(i) ensures the invertibility of the matrix $n^{-1}\sum_{i=1}^n \boldsymbol{X}_i^\top \boldsymbol{\Gamma}_i \boldsymbol{R}^{-1} \boldsymbol{\Gamma}_i \boldsymbol{X}_i$, and Condition (C4)(ii) ensures the applicability of the Lindeberg–Feller central limit theorem. Condition (C5) is also adopted in \cite{Fu12}. Condition (C6) mainly restricts the inverse probability weights, similar to the condition adopted in \cite{Ai21b}.

%%%%%%%%%%%%%%%%%%%%%%%%   Theorem1   %%%%%%%%%%%%%%%%%%%%%%%%
\begin{theorem}\label{theorem:consistency_smoothed}
Under conditions (C1)--(C6), if $n \to \infty$ and $r \to \infty$, then for any $\epsilon > 0$, there exists a finite constant $\Delta_{\epsilon}$ and an integer $r_{\epsilon}$ such that
\[
P\left(\left\|\tilde{\boldsymbol{\beta}} - \boldsymbol{\beta}_0\right\| \ge r^{-1/2}\Delta_{\epsilon}\right) < \epsilon
\]
for all $r > r_{\epsilon}$.
\end{theorem}

%%%%%%%%%%%%%%%%%%%%%%%%   Theorem2   %%%%%%%%%%%%%%%%%%%%%%%%
\begin{theorem}\label{theorem:asymptotic_normality_smoothed}
    Under conditions (C1)--(C6), if $n \to \infty$ and $r \to \infty$, it follows that
\begin{equation*}
\boldsymbol{V}^{-1/2}(\tilde{\boldsymbol{\beta}}-\boldsymbol{\beta}_0)\xrightarrow{d} N(\boldsymbol{0}, \boldsymbol{I}),
\end{equation*}
where $\boldsymbol{V} = \boldsymbol{H}_n^{-1}\boldsymbol{V}_c\boldsymbol{H}_n^{-1}$, $\boldsymbol{H}_n=-n^{-1}\sum_{i=1}^n \boldsymbol{X}_i^\top \boldsymbol{\Gamma}_i \boldsymbol{R}^{-1} \boldsymbol{\Gamma}_i \boldsymbol{X}_i$, and
\[\boldsymbol{V}_c=\frac{1}{n}\sum_{i=1}^{n}\frac{1}{n\pi_i}\boldsymbol{X}_i^\top \boldsymbol{\Gamma}_i \boldsymbol{R}^{-1}\boldsymbol{V}_0\boldsymbol{R}^{-1}\boldsymbol{\Gamma}_i\boldsymbol{X}_i.\]
\end{theorem}

\section{Optimal Poisson subsampling}

In this section, the optimal Poisson subsampling probabilities are derived, and a practical algorithm is proposed for implementation.

\subsection{Optimal Poisson subsampling strategy}

The asymptotic normality result in Theorem \ref{theorem:asymptotic_normality_smoothed} can be used to derive the optimal Poisson subsampling probabilities by minimizing the asymptotic mean squared error (AMSE) of the subsample estimator. This is equivalent to minimizing $\operatorname{tr}(\boldsymbol{V})$, which corresponds to the A-optimality criterion. We present the following theorem.

%%%%%%%%%%%%%%%%%%%%%%%%   Theorem3   %%%%%%%%%%%%%%%%%%%%%%%%
\begin{theorem}\label{theorem:opt}
    For $i = 1,2,\cdots,n$, define $z_i = \left\| \boldsymbol{H}_n^{-1} \boldsymbol{X}_i^\top \boldsymbol{\Gamma}_i \boldsymbol{R}^{-1}\boldsymbol{V}_0^{1/2}\right\|_F$, where $\|\cdot\|_F$ denotes the Frobenius norm, and let $z_{(1)} \leq z_{(2)} \leq \cdots \leq z_{(n)}$ denote the order statistics of $\{z_i\}_{i=1}^n$. If the subsampling probabilities are given by
    \begin{equation}\label{eq:opt_pro}
       \pi_i = r \frac{z_i \wedge T}{\sum_{j=1}^n (z_j \wedge T)},
    \end{equation}
    the value of $\mathrm{tr}\left( \boldsymbol{V}\right)$ is minimized, where
    \[ T = \sum_{i=1}^{n-u} z_{(i)} / (r - u), \]
    and
    \[ u = \min \left\{ v \mid 0 \leq v \leq r, z_{(n-v)} < \sum_{i=1}^{n-v} z_{(i)} / (r - v) \right\}. \]
\end{theorem}

The computation of $T$ in (\ref{eq:opt_pro}) is only necessary when $r z_i / \sum_{j=1}^n z_j > 1$ for some $i$. In such cases, $u$ is defined as the number of indices for which $\pi_i = 1$. Conversely, if $r z_i / \sum_{j=1}^n z_j \leq 1$ for all $i$, then the optimal subsampling probabilities reduce to $\pi_i = r z_i / \sum_{j=1}^n z_j$, and no truncation is needed.

\subsection{Practically implementable algorithm}\label{sec:prac_imple}

Rewriting (\ref{eq:opt_pro}) yields
\[
\pi_i = r \frac{z_i \wedge T}{\sum_{j=1}^n (z_j \wedge T)} = r \frac{z_i \wedge T}{n\Psi},
\]
where $\Psi = n^{-1}\sum_{j=1}^n (z_j \wedge T), i = 1, \cdots, n$. Note that the optimal subsampling probabilities involve unknown parameters. Therefore, a practical algorithm is developed for implementation.

First, a pilot sample with expected size $r_{\mathcal{P}}$ is selected using uniform Poisson subsampling, and the index set of the selected subjects is denoted by $\mathcal{I}_{\mathcal{P}}$. Under a working independence correlation structure, a pilot estimator of $\boldsymbol{\beta}_0$, denoted by $\tilde{\boldsymbol{\beta}}_{\mathcal{P}}$, is obtained. Following Wang and Carey \cite{Wang03}, the working correlation matrix $\boldsymbol{R}$ is estimated via the method of moments, with the resulting estimate denoted by $\bar{\boldsymbol{R}}$. For computational simplicity, $\boldsymbol{\Gamma}_i$ is set to the identity matrix in practical implementation, as in Zu \textit{et al.} \cite{Zu23}. In addition, Yu \textit{et al.} \cite{Yu22} showed through simulation studies that setting $T=\infty$ has little effect when the subsampling rate $r/n$ is sufficiently small. For each $i\in\mathcal{I}_{\mathcal{P}}$, define $\tilde{\boldsymbol{\varepsilon}}_{i\mathcal{P}}=\boldsymbol{Y}_i-\boldsymbol{X}_i\tilde{\boldsymbol{\beta}}_{\mathcal{P}}$, $\boldsymbol{\psi}_{i\mathcal{P}}=\tau-I(\tilde{\boldsymbol{\varepsilon}}_{i\mathcal{P}}\leq\boldsymbol{0})$, and $\bar{\boldsymbol{\psi}}_{\mathcal{P}}=|\mathcal{I}_{\mathcal{P}}|^{-1}\sum_{i\in\mathcal{I}_{\mathcal{P}}}\boldsymbol{\psi}_{i\mathcal{P}}$, where $|\mathcal{I}_{\mathcal{P}}|$ denotes the size of $\mathcal{I}_{\mathcal{P}}$. The matrix $\boldsymbol{V}_0$ is estimated from the pilot sample by $\bar{\boldsymbol{V}}_0=|\mathcal{I}_{\mathcal{P}}|^{-1}\sum_{i\in\mathcal{I}_{\mathcal{P}}}(\boldsymbol{\psi}_{i\mathcal{P}}-\bar{\boldsymbol{\psi}}_{\mathcal{P}})(\boldsymbol{\psi}_{i\mathcal{P}}-\bar{\boldsymbol{\psi}}_{\mathcal{P}})^\top$, and $\bar{\boldsymbol{V}}_0^{1/2}$ is obtained through the eigenvalue decomposition of $\bar{\boldsymbol{V}}_0$. Moreover, $\boldsymbol{H}_n$ is estimated by $\hat{\boldsymbol{H}}_n=-|\mathcal{I}_{\mathcal{P}}|^{-1}\sum_{i\in\mathcal{I}_{\mathcal{P}}}\boldsymbol{X}_i^\top\boldsymbol{\Gamma}_i\bar{\boldsymbol{R}}^{-1}\boldsymbol{\Gamma}_i\boldsymbol{X}_i$. Accordingly, define $\hat{z}_i=\|\hat{\boldsymbol{H}}_n^{-1}\boldsymbol{X}_i^\top\boldsymbol{\Gamma}_i\bar{\boldsymbol{R}}^{-1}\bar{\boldsymbol{V}}_0^{1/2}\|_F$. To reduce the computational cost, the pilot-sample average $\hat{\Psi}=|\mathcal{I}_{\mathcal{P}}|^{-1}\sum_{i\in\mathcal{I}_{\mathcal{P}}}\hat{z}_i$ is used to approximate $\Psi$.

Furthermore, to prevent the optimal Poisson subsampling probabilities from being excessively small, the approach in Ma \textit{et al.} \cite{Ma15} is followed, with these probabilities mixed with uniform subsampling probabilities to achieve more robust results. Specifically, the mixed optimal Poisson subsampling probabilities are expressed as
\begin{equation}
\tilde{\pi}_i = (1 - \rho) \frac{r \left\| \hat{\boldsymbol{H}}^{-1}_n \boldsymbol{X}_i^\top \boldsymbol{\Gamma}_i \bar{\boldsymbol{R}}^{-1}\bar{\boldsymbol{V}}_0^{1/2}\right\|_F}{n \hat{\Psi}} + \rho \frac{r}{n}, \quad i = 1, \cdots, n,
\label{eq:adjusted_subsampling_prob}
\end{equation}
where $\rho \in (0,1)$. Subsequently, define $\pi_i^{\mathrm{opt}}=\tilde{\pi}_i\wedge 1$, $i=1,\cdots,n$. Once the probabilities $\pi_i^{\mathrm{opt}}$ are determined, Poisson subsampling is applied to obtain the optimal subsample.

Let $\mathcal{I}_{\mathcal{S}}$ denote the index set of the optimal subsample. For practical implementation, following Zu et al. \cite{Zu23}, the Newton--Raphson algorithm is adopted. Let $\boldsymbol{\beta}^{(t)}$ represent the parameter estimates obtained at the $t$-th iteration, and define
\begin{equation}\label{eq:Sr_bar}
\bar{\boldsymbol{S}}_r^{\Phi}(\boldsymbol{\beta}^{(t)}) = \frac{1}{n} \sum_{i \in \mathcal{I}_{\mathcal{S}}} \frac{1}{\pi_i^\text{opt}} \boldsymbol{X}_i^\top \boldsymbol{\Gamma}_i \bar{\boldsymbol{R}}^{-1} \left[\Phi\left( \frac{\boldsymbol{Y}_i - \boldsymbol{X}_i \boldsymbol{\beta}^{(t)}}{h(\boldsymbol{X}_i\boldsymbol{X}_i^{\top})^{1/2}} \right) - (1 - \tau)\right],
\end{equation}
\begin{equation}\label{eq:Hn_bar}
\bar{\boldsymbol{H}}_n(\boldsymbol{\beta}^{(t)}) = -\frac{1}{n} \sum_{i \in \mathcal{I}_{\mathcal{S}}} \frac{1}{\pi_i^\text{opt}} \boldsymbol{X}_i^\top \boldsymbol{\Gamma}_i \bar{\boldsymbol{R}}^{-1} \operatorname{diag}\left\{h^{-1}(\boldsymbol{X}_i\boldsymbol{X}_i^{\top})^{-1/2} \phi\left(\frac{\boldsymbol{Y}_{i}-\boldsymbol{X}_{i} \boldsymbol{\beta}^{(t)}}{h(\boldsymbol{X}_i\boldsymbol{X}_i^{\top})^{1/2}}\right)\right\} \boldsymbol{X}_i,
\end{equation}
where $\phi(\cdot)$ stands for the density function of the standard normal distribution. Accordingly, the iterative update formula is given by
\begin{equation}\label{eq:update_formula}
    \boldsymbol{\beta}^{(t+1)} = \boldsymbol{\beta}^{(t)} - \bigl[\bar{\boldsymbol{H}}_n(\boldsymbol{\beta}^{(t)})\bigr]^{-1}\bar{\boldsymbol{S}}_r^{\Phi}(\boldsymbol{\beta}^{(t)}).
\end{equation}
The iterative procedure is terminated when
$\|\boldsymbol{\beta}^{(t+1)}-\boldsymbol{\beta}^{(t)}\|\leq 10^{-4}$,
and the resulting value is taken as the final parameter estimate. The complete implementation procedure is summarized in Algorithm \ref{alg:practical_implement}.

\begin{algorithm}[htbp]
  \caption{Practically implementable estimation algorithm.}
  \label{alg:practical_implement}
  \begin{algorithmic}[1]
    \REQUIRE $r_{\mathcal{P}}, r, n, \rho$, and the maximum number of iterations $M$.
    \STATE \textbf{Step 1: Obtain the pilot sample.}
    \STATE Initialization: $\mathcal{P} = \emptyset$;
    \FOR{$i = 1, \cdots, n$}
        \STATE Generate a Bernoulli random variable $\delta_i \sim \operatorname{Bernoulli}(\pi_i)$ with $\pi_i = r_{\mathcal{P}}/n$;
        \IF{$\delta_i = 1$}
            \STATE Update $\mathcal{P} = \mathcal{P} \cup \{(\boldsymbol{Y}_i, \boldsymbol{X}_i, \pi_i)\}$;
        \ENDIF
    \ENDFOR
    \STATE Compute the pilot estimators $\tilde{\boldsymbol{\beta}}_{\mathcal{P}}$, $\bar{\boldsymbol{R}}$, $\bar{\boldsymbol{V}}_0^{1/2}$, $\hat{\Psi}$, and $\hat{\boldsymbol{H}}_n$ based on $\mathcal{P}$.
    \STATE \textbf{Step 2: Subsampling}
    \STATE Initialization: $\mathcal{S} = \emptyset$;
    \FOR{$i = 1, \cdots, n$}
        \STATE Calculate the optimal subsampling probability $\pi_i^{\text{opt}} = (\tilde{\pi}_i \wedge 1)$ by (\ref{eq:adjusted_subsampling_prob});
        \STATE Generate a Bernoulli random variable $\delta_i \sim \operatorname{Bernoulli}(\pi_i^{\text{opt}})$;
        \IF{$\delta_i = 1$}
            \STATE Update $\mathcal{S} = \mathcal{S} \cup \{(\boldsymbol{Y}_i, \boldsymbol{X}_i, \pi_i^{\text{opt}})\}$;
        \ENDIF
    \ENDFOR
    \STATE \textbf{Step 3: Estimation}
\FOR{$t = 0, \cdots, M - 1$}
    \STATE Update $\boldsymbol{\beta}^{(t+1)}$ according to the Newton iteration formula (\ref{eq:update_formula}).
    \STATE \textbf{If} $\|\boldsymbol{\beta}^{(t+1)} - \boldsymbol{\beta}^{(t)}\| \leq 10^{-4}$, \textbf{break} the iteration.
\ENDFOR
\STATE \textbf{Output}: subsample estimator $\tilde{\boldsymbol{\beta}}$.
  \end{algorithmic}
\end{algorithm}

\section{Regularized estimation based on the optimal subsample}

In this section, we introduce the penalized weighted smoothed estimating functions and develop an implementation algorithm within the optimal Poisson subsampling framework.

\subsection{Penalized weighted smoothed estimating functions}

Regularization-based parameter estimation is essential in high-dimensional data analysis. Under the assumptions of non-massive data and sparsity, numerous results have been established; see, for example, \cite{Fan01, Zou06, Hao16, Zhu21}. However, relatively little work has been devoted to regularized parameter estimation for large-scale longitudinal data. Therefore, we introduce a regularization method within the optimal Poisson subsampling framework. The penalized weighted smoothed estimating functions with the adaptive LASSO (ALASSO) penalty are defined as
\[
\boldsymbol{S}_r^P(\boldsymbol{\beta}) = \boldsymbol{S}_r^\Phi(\boldsymbol{\beta}) - \lambda \tilde{\boldsymbol{w}} \boldsymbol{\operatorname{sgn}}(\boldsymbol{\beta}),
\]
where $\tilde{\boldsymbol{w}} = \left(|\tilde{\beta}_{1}|^{-\gamma}, \cdots, |\tilde{\beta}_{p}|^{-\gamma}\right)^\top$ denotes the adaptive weight vector, with $\gamma>0$. Here, $\tilde{\boldsymbol{\beta}}$ denotes the unpenalized estimator obtained by solving the weighted smoothed estimating equations presented in Section 3, $\lambda>0$ is the tuning parameter, and $\tilde{\boldsymbol{w}} \boldsymbol{\mathrm{sgn}}(\boldsymbol{\beta})$ denotes the element-wise product. In addition, $\boldsymbol{\operatorname{sgn}}(\boldsymbol{\beta})=
\bigl(\operatorname{sgn}(\beta_1),\ldots,\operatorname{sgn}(\beta_p)\bigr)^\top$,
where $\operatorname{sgn}(\beta_k)=I(\beta_k>0)-I(\beta_k<0)$ for $k=1,\cdots,p$. In this paper, we set $\gamma = 1$.

To establish the theoretical properties of regularized parameter estimation, we impose condition (C7).
\begin{itemize} 
 \item[(C7)] Suppose that $\sqrt{r}\lambda \to 0$ and $r\lambda \to \infty$. (i) For $\beta_k\neq 0$ and $|\tilde{w}_k| < \infty$, it follows that $\lim\limits_{r\to\infty}\sqrt{r}\,\lambda\,\tilde{w}_k = 0$. (ii) To obtain sparsity, when $\beta_k$ is sufficiently small, i.e., $|\beta_k| < Mr^{-1/2}$ for some constant $M$, combined with the definition of $\tilde{\boldsymbol{w}}$, this yields $\lim\limits_{r\to\infty} \sqrt{r} \inf_{|\beta_k| \leq M r^{-1/2}} \lambda \tilde{w}_k = \infty$.
\end{itemize}

\begin{theorem}\label{theorem:penalty}
Let the number of non-zero coefficients be $s$. Under conditions (C1)--(C7), the following conclusions hold:
\begin{itemize}
    
    \item[(a)] There exists a $\sqrt{r}$-consistent approximate zero-crossing point of $\boldsymbol{S}_r^P(\boldsymbol{\beta})$, i.e., there exists an estimator $\check{\boldsymbol{\beta}}$ satisfying $\check{\boldsymbol{\beta}} = \boldsymbol{\beta}_0 + O_p(r^{-1/2}),$ such that $\check{\boldsymbol{\beta}}$ is an approximate zero-crossing point of $\boldsymbol{S}_r^P(\boldsymbol{\beta})$.
    \item[(b)] For any $\sqrt{r}$-consistent approximate zero-crossing point of $\boldsymbol{S}_r^P(\boldsymbol{\beta})$, denoted by $\check{\boldsymbol{\beta}} = (\check{\beta}_1, \cdots, \check{\beta}_p)^{\top}$, it holds that
\[
\lim_{r\to\infty} P\big(\check{\beta}_k = 0, \text{for } k>s\big) = 1.
\]
Furthermore, denote $\check{\boldsymbol{\beta}}_1 = (\check{\beta}_1, \cdots, \check{\beta}_s)^{\top}$ and $\boldsymbol{\beta}_{01} = (\beta_{01}, \cdots, \beta_{0s})^{\top}$ as the first $s$ components of $\check{\boldsymbol{\beta}}$ and $\boldsymbol{\beta}_0$, respectively. Then
\[
\boldsymbol{V}_{c11}^{-1/2}\boldsymbol{H}_{n11}\big(\check{\boldsymbol{\beta}}_1 - \boldsymbol{\beta}_{01} + \boldsymbol{H}_{n11}^{-1}\boldsymbol{b}_n\big) \stackrel{d}{\to} N\big(\boldsymbol{0}, \boldsymbol{I}_s\big).
\]
Here, $\boldsymbol{H}_{n11}$ and $\boldsymbol{V}_{c11}$ stand for the leading $s\times s$ submatrices of $\boldsymbol{H}_{n}$ and $\boldsymbol{V_c}$, respectively, and $\boldsymbol{b}_n$ is defined as $-(\lambda \tilde{w}_1\operatorname{sgn}(\beta_{01}),\cdots, \lambda \tilde{w}_s\operatorname{sgn}(\beta_{0s}))^{\top}$.
    \item[(c)] Let $\boldsymbol{S}_{r1}^P(\boldsymbol{\beta})$ and $\boldsymbol{S}_{r1}^\Phi(\boldsymbol{\beta})$ denote the first $s$ components of $\boldsymbol{S}_r^P(\boldsymbol{\beta})$ and $\boldsymbol{S}_r^\Phi(\boldsymbol{\beta})$, respectively. Let $\boldsymbol{\beta}=(\boldsymbol{\beta}_1^\top,\boldsymbol{\beta}_2^\top)^{\top}$, where $\boldsymbol{\beta}_1$ consists of the first $s$ components of $\boldsymbol{\beta}$ and $\boldsymbol{\beta}_2$ consists of the remaining $p-s$ components. Without loss of generality, assume $\boldsymbol{\beta}_2=\boldsymbol{0}$. Then there exists $\check{\boldsymbol{\beta}}_1$ such that \[ \boldsymbol{S}_{r1}^P\left((\check{\boldsymbol{\beta}}_1^\top,\boldsymbol{0}^\top)^\top\right)=\boldsymbol{0}, \] which implies that the solution is exact in this case.
\end{itemize}
\end{theorem}

In the penalized weighted smoothed estimating functions, $\lambda$ serves as a crucial tuning parameter. Following Song \textit{et al.} \cite{Song24} and Li \textit{et al.} \cite{Li24b}, the optimal $\lambda$ is chosen by minimizing the Bayesian information criterion (BIC) defined as
\[
\text{BIC} = \log\left( \frac{1}{n}\sum_{i=1}^n \sum_{j=1}^m \frac{\delta_i}{\pi_i^{\text{opt}}} \rho_\tau\left(y_{ij} - \boldsymbol{x}_{ij}^\top\boldsymbol{\beta}\right) \right) + df \cdot \frac{\log(n)}{2n},
\]
where $\rho_\tau(u)=u[\tau-I(u<0)]$ and $df$ is the number of nonzero coefficients.

\subsection{Algorithm for implementing regularized estimation}

The minorization-maximization algorithm is utilized to address the nonsmoothness induced by the ALASSO penalty. Under the proposed framework, the penalized weighted smoothed estimating functions are approximated as
\[
\boldsymbol{S}_r^P(\boldsymbol{\beta}) = \boldsymbol{S}_r^{\Phi}(\boldsymbol{\beta}) - \lambda \tilde{\boldsymbol{w}} \boldsymbol{\mathrm{sgn}}(\boldsymbol{\beta}) \frac{|\boldsymbol{\beta}|}{\epsilon + |\boldsymbol{\beta}|} = \boldsymbol{0},
\]
where $\epsilon$ represents a small positive constant, which is set to $10^{-6}$ in practical implementation. As in the unpenalized case, let $\boldsymbol{\beta}^{(t)}$ denote the parameter estimates at the $t$-th iteration. Following Zu \textit{et al.} \cite{Zu23}, we define
\[
\boldsymbol{E}\left(\boldsymbol{\beta}^{(t)}\right) = \operatorname{diag}\left\{ \frac{\lambda \tilde{w}_1}{\epsilon + |{\beta}^{(t)}_{1}|}, \cdots, \frac{\lambda \tilde{w}_p}{\epsilon + |{\beta}^{(t)}_{p}|} \right\}.
\]
Parameter updating is carried out via the Newton iteration formula
\begin{equation}\label{eq:pen_update_formula}
  \boldsymbol{\beta}^{(t+1)} = \boldsymbol{\beta}^{(t)} - \bigg[\bar{\boldsymbol{H}}_n(\boldsymbol{\beta}^{(t)})-\boldsymbol{E}(\boldsymbol{\beta}^{(t)})\bigg]^{-1} \bigg[\bar{\boldsymbol{S}}_r^{\Phi}(\boldsymbol{\beta}^{(t)})-\boldsymbol{E}(\boldsymbol{\beta}^{(t)})\boldsymbol{\beta}^{(t)}\bigg],
\end{equation}
where $\bar{\boldsymbol{H}}_n(\boldsymbol{\beta}^{(t)})$ and $\bar{\boldsymbol{S}}_r^{\Phi}(\boldsymbol{\beta}^{(t)})$ have been defined in (\ref{eq:Hn_bar}) and (\ref{eq:Sr_bar}), respectively. The iteration stops once the convergence condition $\|\boldsymbol{\beta}^{(t+1)} - \boldsymbol{\beta}^{(t)}\| \leq 10^{-4}$ is satisfied, and the resulting estimates are taken as the final parameter values. The detailed implementation procedures are presented in Algorithm \ref{alg:PWQGEE}.
\begin{algorithm}[htbp]
  \caption{Algorithm for implementing regularized estimation.}
  \label{alg:PWQGEE}
  \begin{algorithmic}[1]
    \REQUIRE Optimal subsample set $\mathcal{S}$, $\tilde{\boldsymbol{\beta}}$, $\bar{\boldsymbol{R}}$, $\epsilon$, $\lambda$, and the maximum number of iterations $M$.
\FOR{$t = 0, \dots, M - 1$}
    \STATE Update $\boldsymbol{\beta}^{(t+1)}$ using the Newton iteration formula (\ref{eq:pen_update_formula}).
    \STATE \textbf{If} $\|\boldsymbol{\beta}^{(t+1)} - \boldsymbol{\beta}^{(t)}\| \leq 10^{-4}$, \textbf{break} the iteration.
\ENDFOR
    \STATE \textbf{Output}: regularized estimator $\check{\boldsymbol{\beta}}$.
  \end{algorithmic}
\end{algorithm}

\section{Numerical simulations}

Simulation studies are performed to assess the performance of the proposed optimal Poisson subsampling method and regularized parameter estimation. All simulation results are based on 500 independent replications.

\subsection{Parameter estimation}\label{sec:num par est}

In the numerical simulation studies, the response variable is generated according to the following linear model
\[
y_{ij} = \boldsymbol{x}^\top_{ij} \boldsymbol{\beta}_0 + \varepsilon_{ij} + q_\tau, \quad i = 1,\cdots,n, j=1,\cdots,m.
\]
We set $n = 10000$ subjects and $m = 5$ observations per subject. The true regression coefficient vector is $\boldsymbol{\beta}_0 = (1, 1.5,1,1.5,1)^\top$. The shift term $q_\tau$ is introduced to ensure that $P(\varepsilon_{ij}+q_\tau \leq 0) = \tau$. The covariate $\boldsymbol{x}_{ij}$ follows a multivariate $t_2(\boldsymbol{0}, \boldsymbol{\Sigma})$ distribution, where the $(k,l)$-th entry of $\boldsymbol{\Sigma}$ is given by $\Sigma_{kl}=0.5^{|k-l|}$.

The expected pilot sample size is set to $r_{\mathcal{P}}=400$, and three quantile levels $\tau=0.25, 0.5, 0.75$ are adopted for analysis. For the true correlation matrix with the AR(1) structure, two error distributions are considered: the multivariate normal distribution $\boldsymbol{\varepsilon}_i \sim N(\boldsymbol{0}, \boldsymbol{R}_T)$ and the $t_3$ distribution $\boldsymbol{\varepsilon}_i \sim t_3(\boldsymbol{0}, \boldsymbol{R}_T)$, where the correlation coefficient of $\boldsymbol{R}_T$ is fixed at $0.5$. Meanwhile, three types of working correlation structures, AR(1), EX and MA(1), are taken into account. The same setting applies to the case where the true correlation matrix adopts the EX structure. This yields four simulation scenarios.

\begin{itemize}
    \item[(I)] The true correlation matrix follows the AR(1) structure, with the error term $\boldsymbol{\varepsilon}_i \sim N(\boldsymbol{0}, \boldsymbol{R}_T)$. Three typical quantile levels $\tau=0.25,0.5,0.75$ are considered.
    \item[(II)] The true correlation matrix also follows the AR(1) structure, while the error term $\boldsymbol{\varepsilon}_i \sim t_3(\boldsymbol{0}, \boldsymbol{R}_T)$, with the same quantile levels as in (I).
    \item[(III)] The true correlation matrix adopts the EX structure. The error term $\boldsymbol{\varepsilon}_i$ follows the same distribution as that in (I), and the same quantile levels are adopted.
    \item[(IV)] The true correlation matrix adopts the EX structure. The error term $\boldsymbol{\varepsilon}_i$ follows the same distribution as that in (II), and the same quantile levels are adopted.
\end{itemize}

In the figure labels such as AR(1)-EX, the left term denotes the true correlation structure and the right term denotes the working correlation structure. Throughout all figures and tables in this paper, ``Unif" and ``Opt" refer to the uniform Poisson subsampling method and the optimal Poisson subsampling method, respectively. 

Since the optimal Poisson subsampling probabilities depend on the shrinkage parameter $\rho$, we investigate its influence on estimation accuracy. Under Scenarios I and II with the working correlation matrix set to AR(1), we set the expected subsample size to 600 to assess estimation performance under different $\rho$ values, as presented in Figure \ref{fg:AR1_corr_600_rou}. Numerical results indicate that optimal Poisson subsampling consistently outperforms uniform Poisson subsampling, and better performance is observed when $\rho$ ranges from 0.05 to 0.5. Accordingly, we set $\rho$ to 0.15 throughout the subsequent simulation analyses. 

\begin{figure}[htbp] 
\centering
\includegraphics[width=1\linewidth]{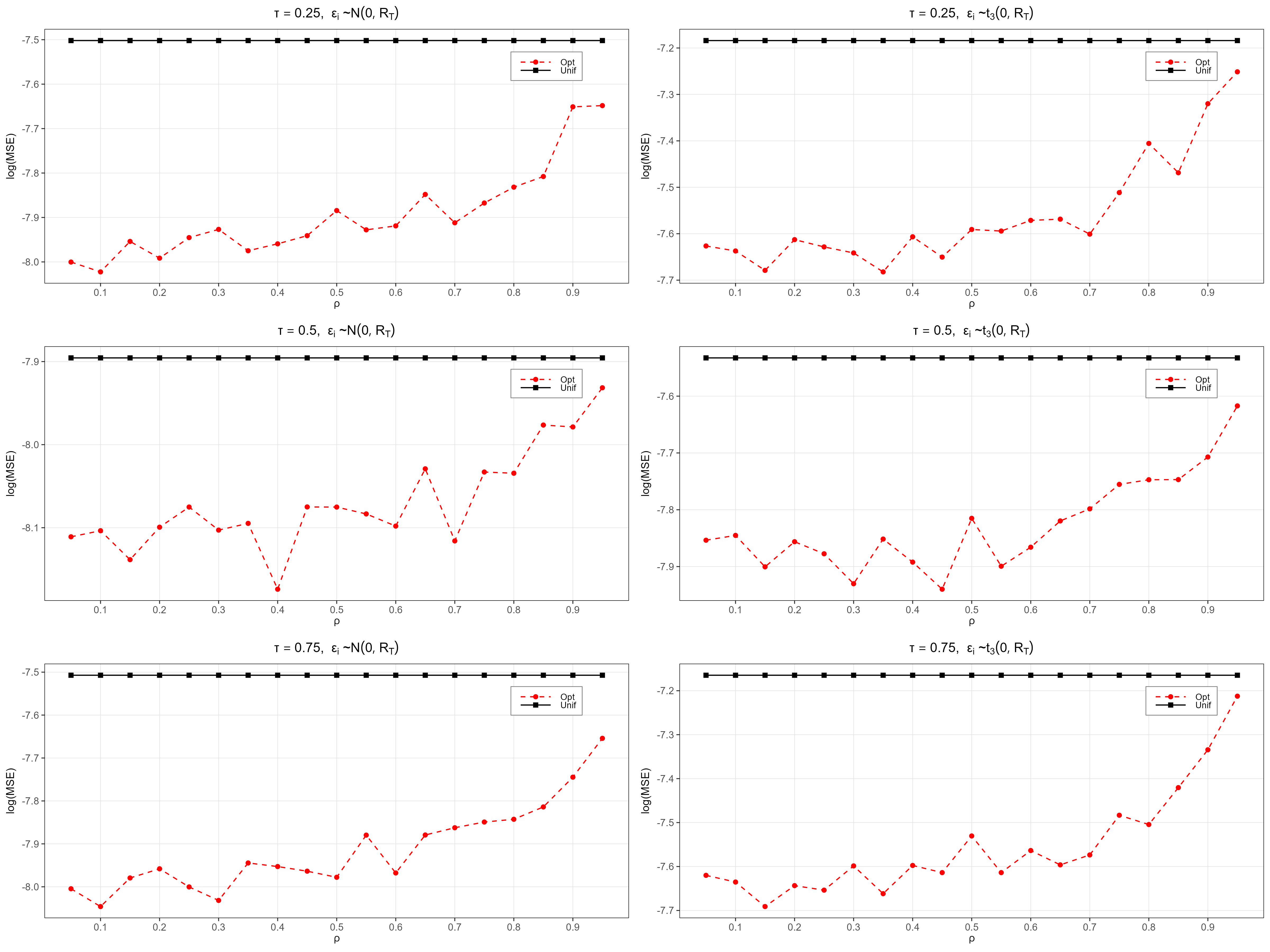}  
\caption{Estimation performance under different $\rho$ values in Scenarios I and II with AR(1) working correlation and $r=600$.}
\label{fg:AR1_corr_600_rou}
\end{figure}

From the log(MSE) results of the two subsampling methods shown in Figures \ref{fg:realAR1_errornorm} to \ref{fg:realEX_errort3}, we draw the following conclusions: the optimal Poisson subsampling method outperforms the uniform Poisson subsampling method, even under misspecification of the working correlation structure. In addition, Tables \ref{tab:bias_sd_tau5_I-II} and \ref{tab:bias_sd_tau5_III-IV} present the biases and standard deviations (SDs) of the two Poisson subsampling methods under the four simulation scenarios with AR(1) and EX working correlation structures at $\tau=0.5$. Similar results are also observed under the other settings. Numerical results show that the biases in all cases are close to 0. The SDs of the optimal Poisson subsampling method are consistently smaller than those of the uniform Poisson subsampling method, and as the subsample size increases, the SDs of both methods exhibit a clear decreasing trend. This demonstrates that the optimal Poisson subsampling method is a more efficient choice for parameter estimation.

\begin{figure}[htbp] 
\centering
\includegraphics[width=1\linewidth]{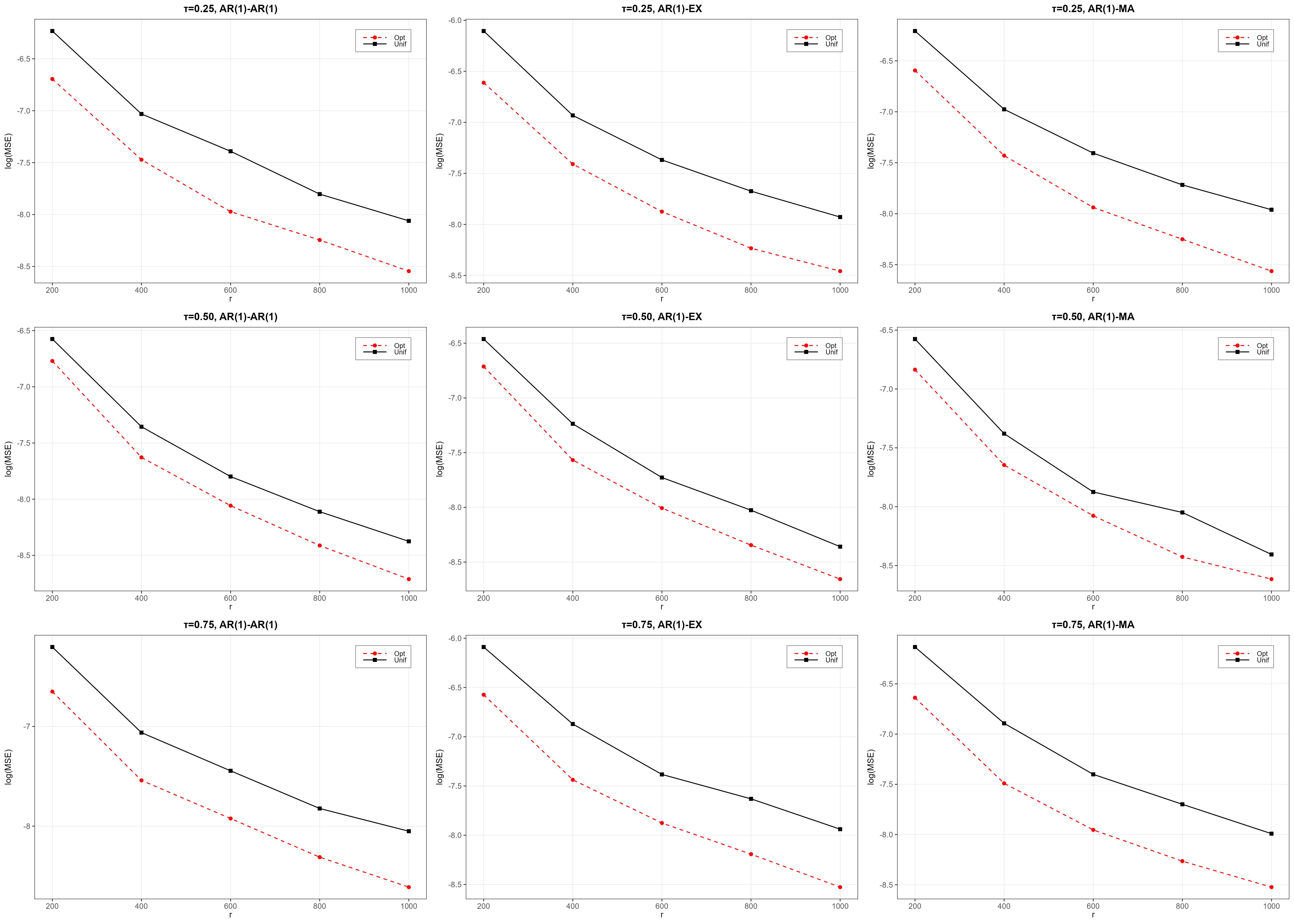}  
\caption{Log(MSE) of the uniform and optimal Poisson subsampling methods under Scenario I.}
\label{fg:realAR1_errornorm}
\end{figure}

\begin{figure}[htbp] 
\centering
\includegraphics[width=1\linewidth]{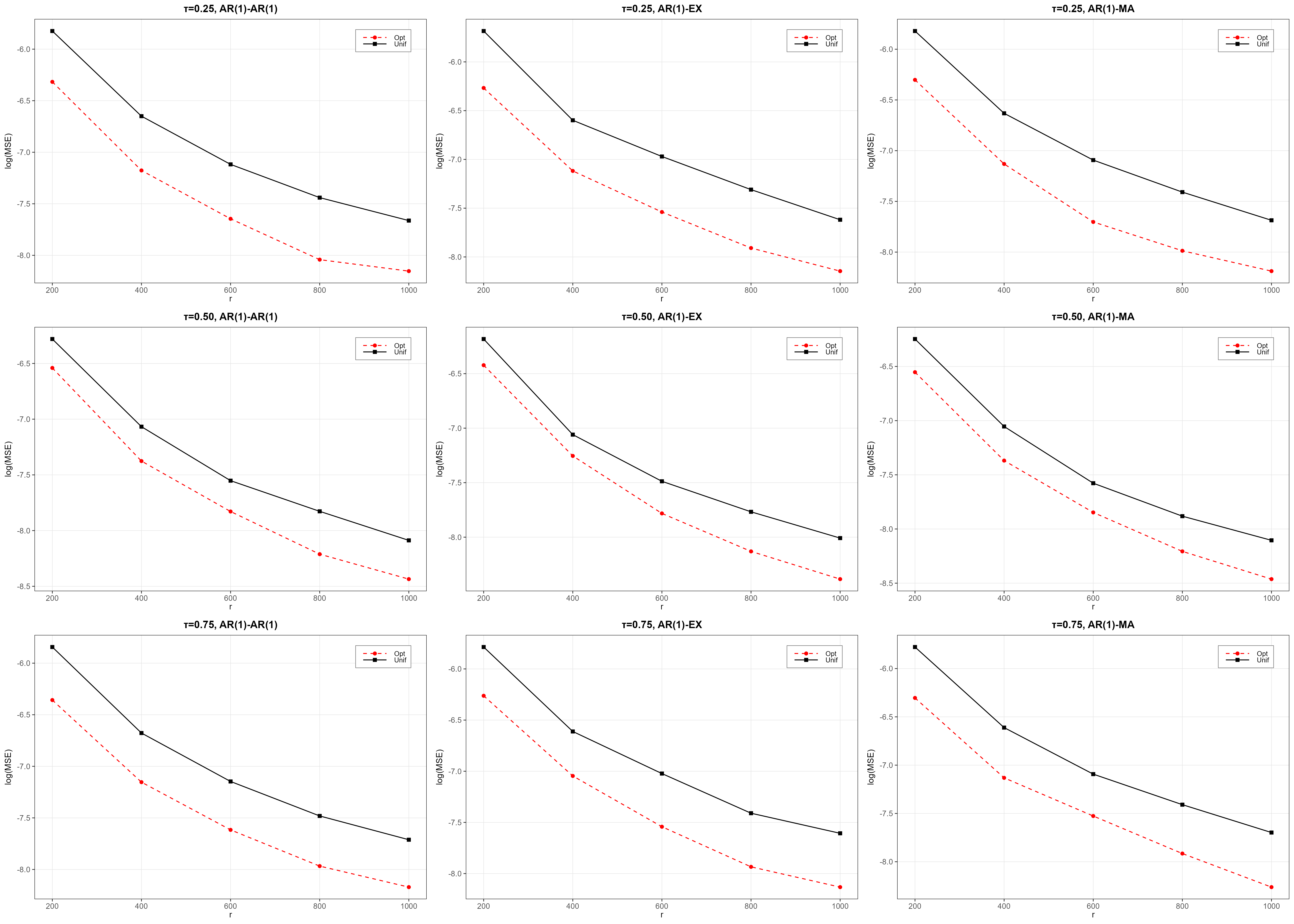}  
\caption{Log(MSE) of the uniform and optimal Poisson subsampling methods under Scenario II.}
\label{fg:realAR1_errort3}
\end{figure}

\begin{figure}[htbp] 
\centering
\includegraphics[width=1\linewidth]{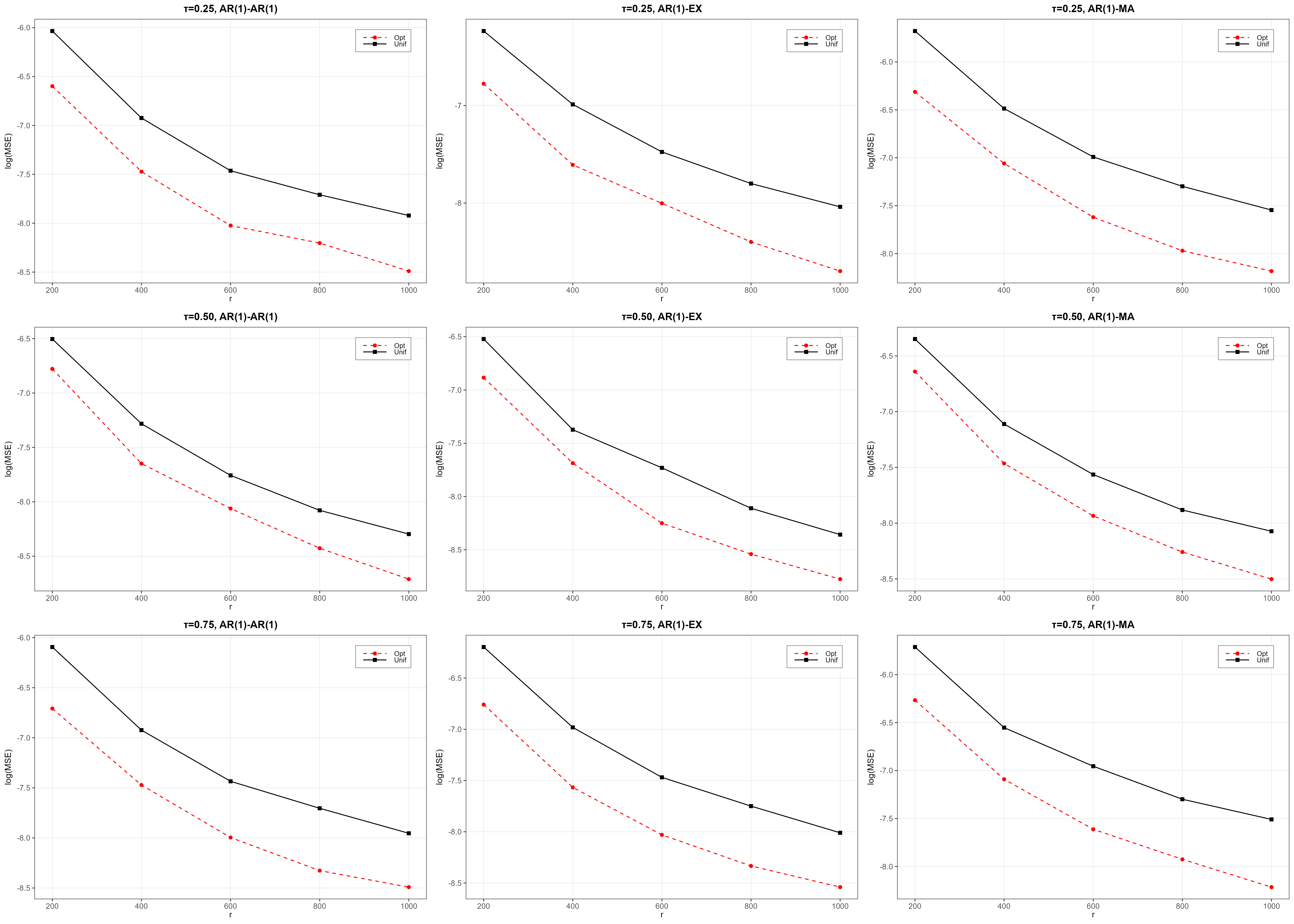}  
\caption{Log(MSE) of the uniform and optimal Poisson subsampling methods under Scenario III.}
\label{fg:realEX_errornorm}
\end{figure}

\begin{figure}[htbp] 
\centering
\includegraphics[width=1\linewidth]{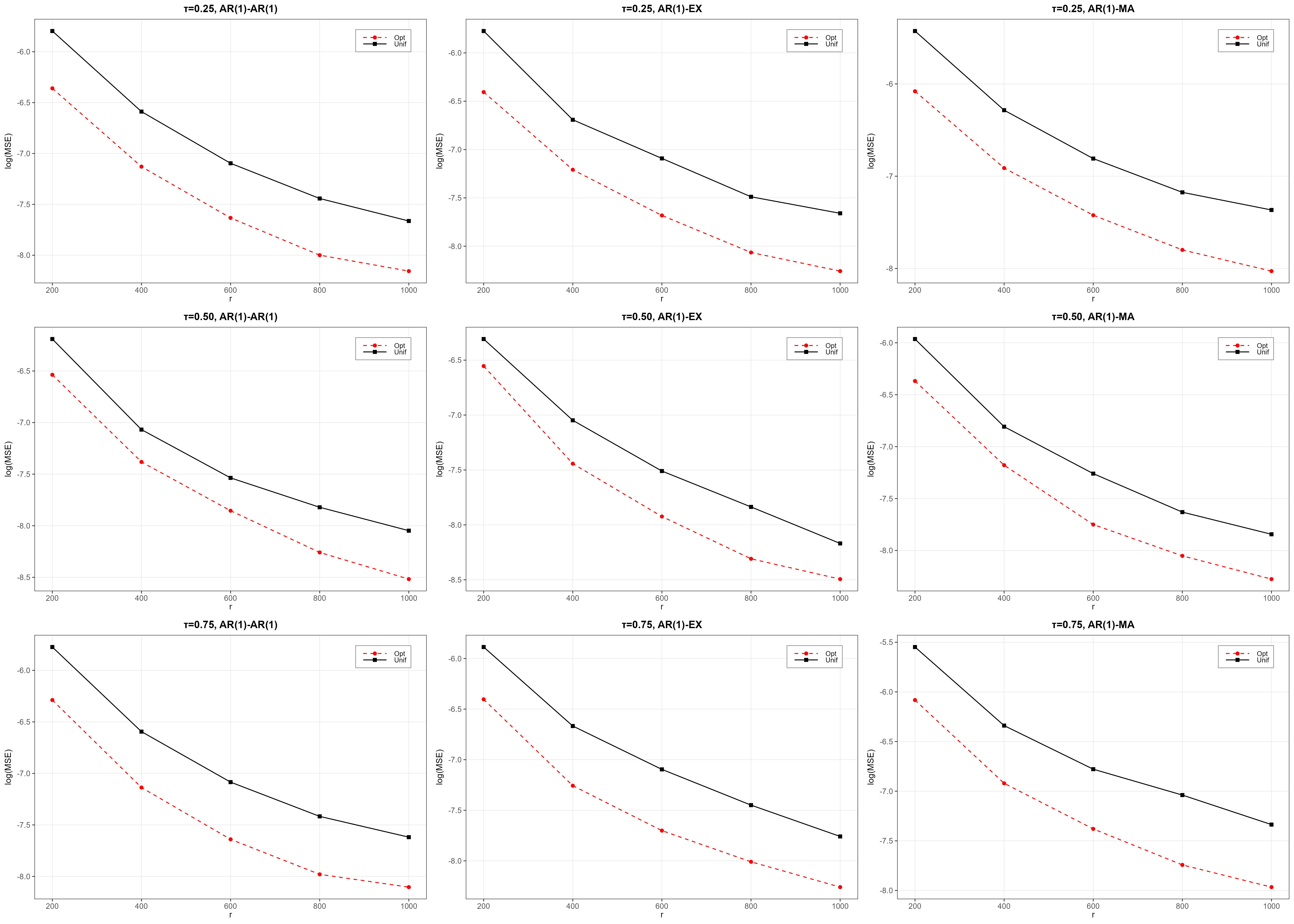}  
\caption{Log(MSE) of the uniform and optimal Poisson subsampling methods under Scenario IV.}
\label{fg:realEX_errort3} 
\end{figure}

\begin{table}[htbp]
  \centering
  \caption{Biases ($\times 10^3$) and Standard Deviations (SDs) of the uniform Poisson subsampling method and the optimal Poisson subsampling method for $\tau = 0.5$ under Scenario I and Scenario II.}
  \label{tab:bias_sd_tau5_I-II}
  \begin{tabular}{cccccccccc}
    \toprule
    \multirow{3}{*}{$r$} & \multirow{3}{*}{Method} & \multicolumn{4}{c}{Scenario I} & \multicolumn{4}{c}{Scenario II} \\
    \cmidrule(lr){3-6} \cmidrule(lr){7-10}
    & & \multicolumn{2}{c}{AR(1)-AR(1)} & \multicolumn{2}{c}{AR(1)-EX} & \multicolumn{2}{c}{AR(1)-AR(1)} & \multicolumn{2}{c}{AR(1)-EX} \\
    & & Bias & SD & Bias & SD & Bias & SD & Bias & SD \\
    \midrule
    200  & Unif & 0.110 & 0.083 & 0.688 & 0.088 & -0.510 & 0.097 & 1.085 & 0.101 \\
         & Opt  & 0.505 & 0.075 & 0.562 & 0.078 & 0.281 & 0.085 & -3.093 & 0.090 \\
    400  & Unif & -0.953 & 0.056 & -0.213 & 0.060 & 0.345 & 0.065 & 0.202 & 0.065 \\
         & Opt  & -0.388 & 0.049 & -0.102 & 0.051 & 0.218 & 0.056 & 0.253 & 0.059 \\
    600  & Unif & 0.408 & 0.045 & 0.126 & 0.047 & -0.128 & 0.051 & 0.153 & 0.053 \\
         & Opt  & 0.204 & 0.040 & 0.055 & 0.041 & -0.234 & 0.045 & -0.669 & 0.046 \\
    800  & Unif & -0.126 & 0.039 & -0.054 & 0.040 & 0.548 & 0.045 & 0.313 & 0.046 \\
         & Opt  & -0.297 & 0.033 & -0.552 & 0.034 & 0.112 & 0.037 & 0.311 & 0.038 \\
    1000 & Unif & -0.628 & 0.034 & -0.213 & 0.034 & 0.009 & 0.039 & 0.668 & 0.041 \\
         & Opt  & 0.048 & 0.029 & 0.380 & 0.029 & 0.470 & 0.033 & -0.020 & 0.034 \\
    \bottomrule
  \end{tabular}
\end{table}

\begin{table}[htbp]
  \centering
  \caption{Biases ($\times 10^3$) and Standard Deviations (SDs) of the uniform Poisson subsampling method and the optimal Poisson subsampling method for $\tau = 0.5$ under Scenario III and Scenario IV.}
  \label{tab:bias_sd_tau5_III-IV}
  \begin{tabular}{cccccccccc}
    \toprule
    \multirow{3}{*}{$r$} & \multirow{3}{*}{Method} & \multicolumn{4}{c}{Scenario III} & \multicolumn{4}{c}{Scenario IV} \\
    \cmidrule(lr){3-6} \cmidrule(lr){7-10}
    & & \multicolumn{2}{c}{EX-AR(1)} & \multicolumn{2}{c}{EX-EX} & \multicolumn{2}{c}{EX-AR(1)} & \multicolumn{2}{c}{EX-EX} \\
    & & Bias & SD & Bias & SD & Bias & SD & Bias & SD \\
    \midrule
    200  & Unif & -0.298 & 0.086 & 2.784 & 0.085 & -0.147 & 0.101 & -1.149 & 0.095 \\
         & Opt  & -1.560 & 0.075 & -1.196 & 0.072 & 0.120 & 0.085 & 1.676 & 0.084 \\
    400  & Unif & 0.523 & 0.059 & -0.238 & 0.056 & 0.304 & 0.065 & 0.889 & 0.066 \\
         & Opt  & 1.310 & 0.049 & 0.840 & 0.048 & 0.360 & 0.056 & 0.850 & 0.054 \\
    600  & Unif & 0.260 & 0.046 & 0.059 & 0.047 & -0.226 & 0.052 & 0.358 & 0.052 \\
         & Opt  & 0.384 & 0.040 & 0.306 & 0.036 & 1.002 & 0.044 & 0.315 & 0.042 \\
    800  & Unif & -0.128 & 0.039 & 0.181 & 0.039 & 0.011 & 0.045 & 0.316 & 0.044 \\
         & Opt  & -0.484 & 0.033 & 0.001 & 0.031 & -0.016 & 0.036 & 0.719 & 0.035 \\
    1000 & Unif & 0.183 & 0.035 & -0.047 & 0.034 & 0.379 & 0.040 & -0.473 & 0.038 \\
         & Opt  & 0.265 & 0.029 & -0.181 & 0.028 & 0.190 & 0.032 & 0.364 & 0.032 \\
    \bottomrule
  \end{tabular}
\end{table}

\subsection{Regularized estimation}

This subsection conducts simulation analyses to evaluate the performance of the proposed method for regularized parameter estimation. Under the four scenarios, the variable dimension is set to $p=10$. The true coefficient vector is set as $\boldsymbol{\beta}_0=(1, 1.5, 1, 1.5, 1, 0, \cdots, 0)^\top$, that is, the first 5 components of $\boldsymbol{\beta}_0$ are 1, 1.5, 1, 1.5, 1, and the remaining 5 components are all 0. All other simulation settings remain identical to those in Section~\ref{sec:num par est}.

The metrics of MSE, Sensitivity, Specificity, and Correct Classification Rate (CCR) are adopted to evaluate the performance of the regularized parameter estimation. Sensitivity refers to the proportion of true non-zero coefficients that are correctly identified as non-zero, while Specificity refers to the proportion of true zero coefficients that are correctly identified as zero. In particular, CCR is defined as the proportion of replications in which all nonzero and zero coefficients are correctly identified. The closer the values of Sensitivity, Specificity, and CCR are to 1, the better the variable selection performance.

Table \ref{tab:simulation_results_scadI_II} presents the regularized parameter estimation results under Scenario I where the error term follows a multivariate normal distribution and Scenario II where the error term follows a $t_3$ distribution, with the true correlation matrix following an AR(1) structure. The simulation results show that as the optimal Poisson subsample size increases, the MSE decreases gradually, the sensitivity is consistently 1, the specificity approaches 1, and the CCR also achieves satisfactory performance. Table \ref{tab:simulation_results_scadIII_IV} reports the regularized parameter estimation results under Scenarios III and IV, which are similar to those in Table \ref{tab:simulation_results_scadI_II}. Therefore, it can be concluded that the proposed regularized parameter estimation method under the optimal Poisson subsampling framework performs favorably.

\begin{sidewaystable}[htbp]
\centering
\setlength{\tabcolsep}{2pt}  % 缩小列间距
\footnotesize
\caption{MSE ($\times 10^3$), sensitivity, specificity, and correct classification rate (CCR) for regularized estimation based on optimal Poisson subsamples under Scenarios I and II.}
\label{tab:simulation_results_scadI_II}
\begin{tabular}{cccccccccccc}
\toprule
\multirow{2}{*}{Scenario} & \multirow{2}{*}{$r$} & \multirow{2}{*}{Metric} 
& \multicolumn{3}{c}{$\tau=0.25$} & \multicolumn{3}{c}{$\tau=0.50$} & \multicolumn{3}{c}{$\tau=0.75$} \\
\cmidrule(lr){4-6} \cmidrule(lr){7-9} \cmidrule(lr){10-12}
& & & AR(1)-AR(1) & AR(1)-EX & AR(1)-MA(1) & AR(1)-AR(1) & AR(1)-EX & AR(1)-MA(1) & AR(1)-AR(1) & AR(1)-EX & AR(1)-MA(1) \\
\midrule
\multirow{16}{*}{I}
& \multirow{4}{*}{400} 
& MSE & 0.635 & 0.758 & 0.676 & 0.575 & 0.645 & 0.557 & 0.658 & 0.726 & 0.691 \\
& & Sensitivity & 1.000 & 1.000 & 1.000 & 1.000 & 1.000 & 1.000 & 1.000 & 1.000 & 1.000 \\
& & Specificity & 0.965 & 0.942 & 0.970 & 0.976 & 0.982 & 0.988 & 0.959 & 0.961 & 0.973 \\
& & CCR & 0.946 & 0.918 & 0.946 & 0.960 & 0.970 & 0.976 & 0.944 & 0.936 & 0.960 \\
\cmidrule(lr){2-12}
& \multirow{4}{*}{600}
& MSE & 0.367 & 0.420 & 0.430 & 0.340 & 0.369 & 0.338 & 0.380 & 0.433 & 0.400 \\
& & Sensitivity & 1.000 & 1.000 & 1.000 & 1.000 & 1.000 & 1.000 & 1.000 & 1.000 & 1.000 \\
& & Specificity & 0.994 & 0.990 & 0.991 & 0.995 & 0.998 & 0.995 & 0.994 & 0.992 & 0.992 \\
& & CCR & 0.986 & 0.980 & 0.984 & 0.992 & 0.994 & 0.990 & 0.986 & 0.976 & 0.982 \\
\cmidrule(lr){2-12}
& \multirow{4}{*}{800}
& MSE & 0.269 & 0.285 & 0.331 & 0.232 & 0.263 & 0.250 & 0.271 & 0.294 & 0.303 \\
& & Sensitivity & 1.000 & 1.000 & 1.000 & 1.000 & 1.000 & 1.000 & 1.000 & 1.000 & 1.000 \\
& & Specificity & 0.999 & 0.999 & 0.997 & 1.000 & 0.998 & 1.000 & 0.996 & 0.999 & 0.999 \\
& & CCR & 0.998 & 0.996 & 0.990 & 1.000 & 0.996 & 0.998 & 0.994 & 0.996 & 0.994 \\
\cmidrule(lr){2-12}
& \multirow{4}{*}{1000}
& MSE & 0.220 & 0.236 & 0.238 & 0.182 & 0.207 & 0.186 & 0.218 & 0.243 & 0.226 \\
& & Sensitivity & 1.000 & 1.000 & 1.000 & 1.000 & 1.000 & 1.000 & 1.000 & 1.000 & 1.000 \\
& & Specificity & 1.000 & 1.000 & 1.000 & 1.000 & 0.998 & 1.000 & 0.998 & 0.999 & 0.999 \\
& & CCR & 0.998 & 0.998 & 0.998 & 1.000 & 0.998 & 0.998 & 0.998 & 0.996 & 0.996 \\
\midrule
\multirow{16}{*}{II}
& \multirow{4}{*}{400}
& MSE & 0.907 & 0.989 & 0.955 & 0.694 & 0.733 & 0.678 & 0.917 & 0.953 & 0.967 \\
& & Sensitivity & 1.000 & 1.000 & 1.000 & 1.000 & 1.000 & 1.000 & 1.000 & 1.000 & 1.000 \\
& & Specificity & 0.980 & 0.965 & 0.982 & 0.984 & 0.991 & 0.986 & 0.963 & 0.970 & 0.957 \\
& & CCR & 0.964 & 0.930 & 0.970 & 0.972 & 0.984 & 0.970 & 0.944 & 0.950 & 0.932 \\
\cmidrule(lr){2-12}
& \multirow{4}{*}{600}
& MSE & 0.580 & 0.606 & 0.549 & 0.437 & 0.473 & 0.432 & 0.563 & 0.581 & 0.575 \\
& & Sensitivity & 1.000 & 1.000 & 1.000 & 1.000 & 1.000 & 1.000 & 1.000 & 1.000 & 1.000 \\
& & Specificity & 0.995 & 0.995 & 0.995 & 0.998 & 0.997 & 1.000 & 0.997 & 0.994 & 0.996 \\
& & CCR & 0.986 & 0.982 & 0.988 & 0.990 & 0.994 & 1.000 & 0.990 & 0.984 & 0.986 \\
\cmidrule(lr){2-12}
& \multirow{4}{*}{800}
& MSE & 0.391 & 0.439 & 0.388 & 0.299 & 0.328 & 0.321 & 0.386 & 0.418 & 0.414 \\
& & Sensitivity & 1.000 & 1.000 & 1.000 & 1.000 & 1.000 & 1.000 & 1.000 & 1.000 & 1.000 \\
& & Specificity & 0.998 & 0.997 & 1.000 & 1.000 & 1.000 & 1.000 & 1.000 & 0.998 & 0.998 \\
& & CCR & 0.992 & 0.992 & 0.998 & 1.000 & 0.998 & 0.998 & 1.000 & 0.994 & 0.998 \\
\cmidrule(lr){2-12}
& \multirow{4}{*}{1000}
& MSE & 0.305 & 0.329 & 0.308 & 0.232 & 0.258 & 0.219 & 0.315 & 0.322 & 0.300 \\
& & Sensitivity & 1.000 & 1.000 & 1.000 & 1.000 & 1.000 & 1.000 & 1.000 & 1.000 & 1.000 \\
& & Specificity & 1.000 & 0.998 & 1.000 & 1.000 & 0.999 & 1.000 & 0.999 & 0.998 & 1.000 \\
& & CCR & 1.000 & 0.994 & 1.000 & 0.998 & 0.996 & 0.998 & 0.996 & 0.998 & 1.000 \\
\bottomrule
\end{tabular}
\end{sidewaystable}

\begin{sidewaystable}[htbp]
\centering
\footnotesize
\caption{MSE ($\times 10^3$), sensitivity, specificity, and correct classification rate (CCR) for regularized estimation based on optimal Poisson subsamples under Scenarios III and IV.}
\label{tab:simulation_results_scadIII_IV}
\begin{tabular}{cccccccccccc}
\toprule
\multirow{2}{*}{Scenario} & \multirow{2}{*}{$r$} & \multirow{2}{*}{Metric}
& \multicolumn{3}{c}{$\tau=0.25$} & \multicolumn{3}{c}{$\tau=0.50$} & \multicolumn{3}{c}{$\tau=0.75$} \\
\cmidrule(lr){4-6} \cmidrule(lr){7-9} \cmidrule(lr){10-12}
& & & EX-AR(1) & EX-EX & EX-MA(1) & EX-AR(1) & EX-EX & EX-MA(1) & EX-AR(1) & EX-EX & EX-MA(1) \\
\midrule
\multirow{16}{*}{III}
& \multirow{4}{*}{400}
& MSE & 0.679 & 0.664 & 0.902 & 0.579 & 0.531 & 0.656 & 0.682 & 0.597 & 0.899 \\
& & Sensitivity & 1.000 & 1.000 & 1.000 & 1.000 & 1.000 & 1.000 & 1.000 & 1.000 & 1.000 \\
& & Specificity & 0.979 & 0.970 & 0.971 & 0.983 & 0.985 & 0.974 & 0.966 & 0.966 & 0.977 \\
& & CCR & 0.966 & 0.952 & 0.940 & 0.978 & 0.976 & 0.964 & 0.944 & 0.948 & 0.944 \\
\cmidrule(lr){2-12}
& \multirow{4}{*}{600}
& MSE & 0.403 & 0.376 & 0.540 & 0.347 & 0.332 & 0.377 & 0.408 & 0.369 & 0.585 \\
& & Sensitivity & 1.000 & 1.000 & 1.000 & 1.000 & 1.000 & 1.000 & 1.000 & 1.000 & 1.000 \\
& & Specificity & 0.996 & 0.999 & 0.998 & 0.996 & 0.996 & 0.999 & 0.993 & 0.995 & 0.994 \\
& & CCR & 0.986 & 0.996 & 0.990 & 0.992 & 0.990 & 0.996 & 0.986 & 0.984 & 0.984 \\
\cmidrule(lr){2-12}
& \multirow{4}{*}{800}
& MSE & 0.294 & 0.267 & 0.406 & 0.252 & 0.215 & 0.278 & 0.296 & 0.270 & 0.390 \\
& & Sensitivity & 1.000 & 1.000 & 1.000 & 1.000 & 1.000 & 1.000 & 1.000 & 1.000 & 1.000 \\
& & Specificity & 0.999 & 0.999 & 0.998 & 0.998 & 1.000 & 0.999 & 1.000 & 0.998 & 0.999 \\
& & CCR & 0.996 & 0.996 & 0.992 & 0.998 & 1.000 & 0.996 & 0.998 & 0.996 & 0.994 \\
\cmidrule(lr){2-12}
& \multirow{4}{*}{1000}
& MSE & 0.241 & 0.214 & 0.308 & 0.187 & 0.185 & 0.204 & 0.234 & 0.197 & 0.314 \\
& & Sensitivity & 1.000 & 1.000 & 1.000 & 1.000 & 1.000 & 1.000 & 1.000 & 1.000 & 1.000 \\
& & Specificity & 0.999 & 1.000 & 0.999 & 0.998 & 1.000 & 1.000 & 0.997 & 0.998 & 1.000 \\
& & CCR & 0.996 & 0.998 & 0.994 & 0.998 & 1.000 & 0.998 & 0.996 & 0.994 & 0.998 \\
\midrule
\multirow{16}{*}{IV}
& \multirow{4}{*}{400}
& MSE & 0.924 & 0.847 & 1.066 & 0.739 & 0.702 & 0.841 & 0.881 & 0.850 & 1.144 \\
& & Sensitivity & 1.000 & 1.000 & 1.000 & 1.000 & 1.000 & 1.000 & 1.000 & 1.000 & 1.000 \\
& & Specificity & 0.984 & 0.980 & 0.970 & 0.986 & 0.992 & 0.984 & 0.983 & 0.979 & 0.971 \\
& & CCR & 0.966 & 0.962 & 0.952 & 0.970 & 0.978 & 0.962 & 0.964 & 0.958 & 0.946 \\
\cmidrule(lr){2-12}
& \multirow{4}{*}{600}
& MSE & 0.535 & 0.503 & 0.669 & 0.429 & 0.388 & 0.517 & 0.547 & 0.516 & 0.685 \\
& & Sensitivity & 1.000 & 1.000 & 1.000 & 1.000 & 1.000 & 1.000 & 1.000 & 1.000 & 1.000 \\
& & Specificity & 0.995 & 0.998 & 0.991 & 0.999 & 0.997 & 0.999 & 0.999 & 0.998 & 0.996 \\
& & CCR & 0.988 & 0.990 & 0.980 & 0.996 & 0.990 & 0.994 & 0.996 & 0.990 & 0.984 \\
\cmidrule(lr){2-12}
& \multirow{4}{*}{800}
& MSE & 0.412 & 0.360 & 0.506 & 0.304 & 0.288 & 0.383 & 0.391 & 0.374 & 0.482 \\
& & Sensitivity & 1.000 & 1.000 & 1.000 & 1.000 & 1.000 & 1.000 & 1.000 & 1.000 & 1.000 \\
& & Specificity & 0.999 & 0.999 & 0.998 & 1.000 & 0.998 & 1.000 & 0.999 & 0.998 & 0.998 \\
& & CCR & 0.996 & 0.994 & 0.994 & 0.998 & 0.998 & 0.998 & 0.996 & 0.992 & 0.990 \\
\cmidrule(lr){2-12}
& \multirow{4}{*}{1000}
& MSE & 0.302 & 0.287 & 0.381 & 0.229 & 0.226 & 0.292 & 0.315 & 0.270 & 0.377 \\
& & Sensitivity & 1.000 & 1.000 & 1.000 & 1.000 & 1.000 & 1.000 & 1.000 & 1.000 & 1.000 \\
& & Specificity & 1.000 & 1.000 & 0.998 & 1.000 & 0.998 & 1.000 & 1.000 & 0.998 & 1.000 \\
& & CCR & 1.000 & 1.000 & 0.996 & 1.000 & 0.998 & 1.000 & 1.000 & 0.996 & 1.000 \\
\bottomrule
\end{tabular}
\end{sidewaystable}

\section{Real data application}

The China Health and Retirement Longitudinal Study (CHARLS) dataset is utilized for empirical validation of the proposed method and is available at http://charls.pku.edu.cn/. Numerous studies have been conducted using this publicly available dataset; see, for example, \cite{Wang24,Zhao26}. Specifically, the 2013, 2015, 2018, and 2020 waves of the CHARLS data are adopted in this study to analyze depression scores. The depression score is used as the dependent variable, and the independent variables include gender, education, and other covariates. Detailed descriptions of these variables are provided in Table \ref{tab:variable_definitions}. After excluding subjects with missing values, the analysis sample consists of 6572 subjects, each with four observations. The depression score, sleep duration, and total cognition score are standardized before analysis. The regression model constructed is as follows:

\[
\begin{split}
\text{depression score} &= \text{Intercept} + \beta_1 \text{gender} + \beta_2 \text{sleep duration} + \beta_3 \text{rural} + \beta_4 \text{total cognition score} \\
 &\quad+ \beta_5 \text{education} +\beta_6 \text{self-rated health}+\beta_7 \text{social participation}\\
 &\quad+\beta_8 \text{pension}+\beta_9 \text{drinking status}.
\end{split}
\]

\begin{table}
\caption{Variable definitions and descriptions.}
{\begin{tabular*}{\linewidth}{@{\extracolsep{\fill}} p{2.8cm} p{\dimexpr\linewidth-3.8cm\relax} @{}}
    \toprule
    Variable name   & Description \\
    \midrule
    gender          & Binary variable: 0 indicates female, 1 indicates male \\
    sleep duration    & nighttime sleep duration of participants \\
    rural           & Binary variable: 0 indicates urban areas, 1 indicates rural areas \\
    total cognition score & The cognition score ranges from 0 to 21, with higher values indicating stronger cognitive ability \\
    education   & Binary variable: 1 indicates education below junior high school, 0 indicates high school or above \\
    self-rated health  & Categorical variable: 1 indicates very poor health, 5 indicates very good health, 2–4 indicate intermediate levels \\
    social participation & Binary variable: 0 indicates no monthly social participation, 1 indicates monthly social participation \\
    pension         & Binary variable: 0 indicates no pension, 1 indicates having pension \\
    drinking status & Binary variable: 0 indicates no alcohol consumption, 1 indicates alcohol consumption \\
    depression score  & Depression score ranges from 0 to 30, with higher scores indicating more severe depression \\
    \bottomrule
\end{tabular*}}
\label{tab:variable_definitions}
\end{table}

Since the true regression parameter vector $\boldsymbol{\beta}_0$ is unknown in empirical analysis, we treat the parameter estimates obtained from the full sample as the true regression parameter $\boldsymbol{\beta}_0$. Figure \ref{fg:realdata_mse} presents the log(MSE) results of the optimal Poisson subsampling method and the uniform Poisson subsampling method at $\tau=0.75$, and the results indicate that the optimal Poisson subsampling method consistently outperforms the uniform Poisson subsampling method. 

\begin{figure}[htbp] 
\centering
\includegraphics[width=1\linewidth]{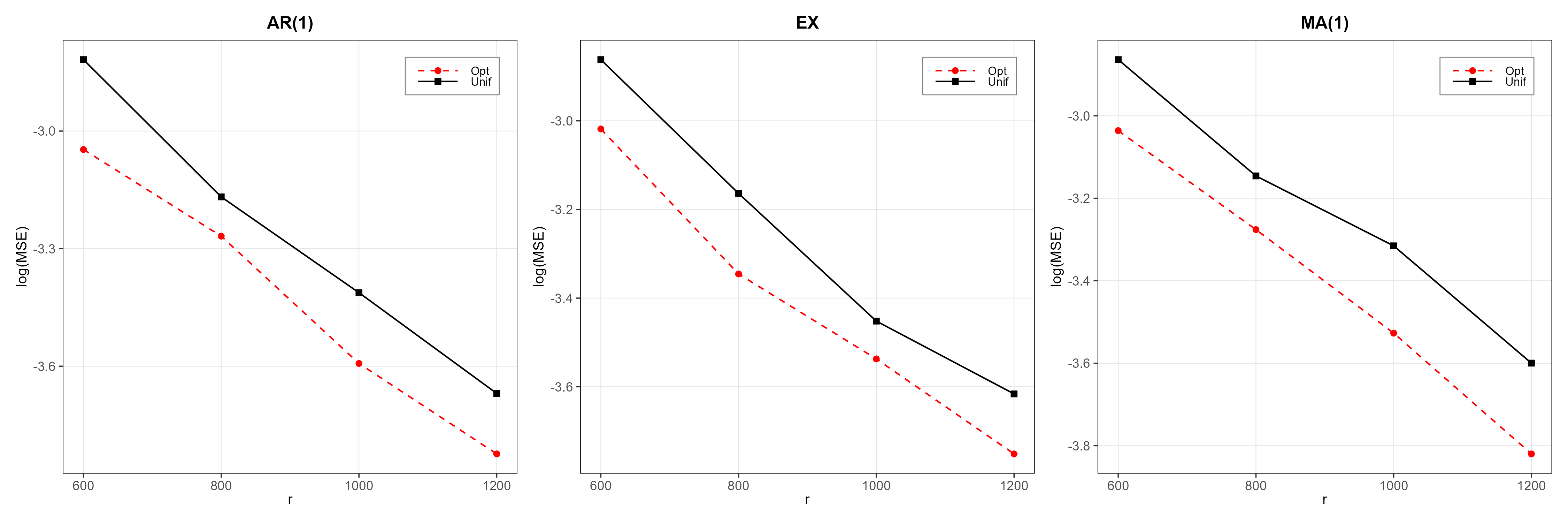}  
\caption{Log(MSE) of the uniform and the optimal Poisson subsampling methods on the real dataset for $\tau=0.75$.}
\label{fg:realdata_mse} 
\end{figure}

In addition, the regularized parameter estimation results for both the full sample and the optimal subsample are presented in Table \ref{tab:realdata_vary}. For the optimal Poisson subsampling method, 500 replications are conducted, with an optimal subsample drawn and regularized estimation performed in each replication. Following Li \textit{et al.} \cite{Li24b}, the median of the 500 coefficient estimates is reported for each variable. In Table \ref{tab:realdata_vary}, ``---'' denotes a coefficient estimate of zero. Specifically, across different working correlation matrices, the variables \text{pension} and \text{drinking status} are shrunk to zero in both the full sample and the optimal subsample, indicating that these variables have relatively weak associations with \text{depression score}.

\begin{table}[htbp]
  \centering
  \caption{Regularized estimation results for the full sample and optimal Poisson subsamples at $\tau = 0.75$.}
  \label{tab:realdata_vary}
  \footnotesize
  \begin{tabular}{ccccccc}
    \toprule
    Working correlation matrix & variable & Full sample & 600 & 800 & 1000 & 1200 \\
    \midrule
    \multirow{10}{*}{AR(1)} 
        & Intercept & 1.475 & 1.456 & 1.468 & 1.469 & 1.467 \\
        & gender & -0.323 & -0.309 & -0.308 & -0.310 & -0.314 \\
        & sleep duration & -0.181 & -0.184 & -0.186 & -0.183 & -0.182 \\
        & rural & 0.229 & 0.223 & 0.219 & 0.229 & 0.223 \\
        & total cognition score & -0.210 & -0.203 & -0.209 & -0.211 & -0.210 \\
        & education & 0.085 & 0.089 & 0.091 & 0.090 & 0.091 \\
        & self-rated health & -0.320 & -0.307 & -0.312 & -0.314 & -0.314 \\
        & social participation & -0.058 & -0.060 & -0.059 & -0.058 & -0.058 \\
        & pension & --- & --- & --- & --- & --- \\
        & drinking status & --- & --- & --- & --- & --- \\
    \midrule
    \multirow{10}{*}{EX} 
        & Intercept & 1.415 & 1.402 & 1.426 & 1.412 & 1.421 \\
        & gender & -0.334 & -0.309 & -0.317 & -0.323 & -0.326 \\
        & sleep duration & -0.174 & -0.177 & -0.177 & -0.174 & -0.177 \\
        & rural & 0.232 & 0.230 & 0.222 & 0.230 & 0.232 \\
        & total cognition score & -0.213 & -0.209 & -0.213 & -0.208 & -0.213 \\
        & education & 0.093 & 0.086 & 0.085 & 0.091 & 0.083 \\
        & self-rated health & -0.302 & -0.294 & -0.298 & -0.298 & -0.298 \\
        & social participation & -0.061 & -0.057 & -0.059 & -0.059 & -0.054 \\
        & pension & --- & --- & --- & --- & --- \\
        & drinking status & --- & --- & --- & --- & --- \\
    \midrule
    \multirow{10}{*}{MA(1)} 
        & Intercept & 1.494 & 1.485 & 1.494 & 1.490 & 1.497 \\
        & gender & -0.320 & -0.302 & -0.307 & -0.300 & -0.301 \\
        & sleep duration & -0.185 & -0.188 & -0.186 & -0.190 & -0.187 \\
        & rural & 0.229 & 0.227 & 0.227 & 0.222 & 0.224 \\
        & total cognition score & -0.208 & -0.207 & -0.207 & -0.207 & -0.210 \\
        & education & 0.089 & 0.090 & 0.079 & 0.086 & 0.087 \\
        & self-rated health & -0.327 & -0.319 & -0.317 & -0.322 & -0.323 \\
        & social participation & -0.061 & -0.059 & -0.063 & -0.064 & -0.059 \\
        & pension & --- & --- & --- & --- & --- \\
        & drinking status & --- & --- & --- & --- & --- \\
    \bottomrule
  \end{tabular}
\end{table}

\section{Conclusion}

The optimal Poisson subsampling algorithm is developed for quantile regression with longitudinal data, and the asymptotic properties of the subsample-based estimator are established. Furthermore, regularized parameter estimation is conducted within the optimal Poisson subsampling framework. As distributed algorithms provide an important means of reducing computational cost for large-scale data, future work will focus on investigating distributed statistical algorithms for large-scale longitudinal data.

\section*{Acknowledgements}

The authors are grateful to all individuals who have contributed to this work.

\section*{Disclosure statement}

No potential conflict of interest was reported by the author(s).

\section*{Funding}

Chunjing Li gratefully acknowledges support from the National Social Science Fund of China (24BTJ061). Xiaohui Yuan acknowledges support from the Scientific Research Project of the Jilin Provincial Department of Science and Technology (20250102029JC).

\bibliographystyle{tfs} 
\bibliography{interacttfssample}

\section*{Auxiliary lemmas and the proof of theorems}

In this section, we provide the proofs of the theorems.
%%%%%%%%%%%%%%%%%%%%%%%%%%%%%%lemma1
\begin{lemma}\label{lemma:Sr}
 Under conditions (C4) and (C6), if $n \to \infty$ and $r \to \infty$, then
$$\boldsymbol{S}_r\left(\boldsymbol{\beta}_{0}\right)=O_p(r^{-1/2}).$$
\end{lemma}
\begin{proof}[Proof of Lemma \ref{lemma:Sr}]  

Direct calculation shows that
\[
\operatorname{E}\bigl[\boldsymbol{S}_r(\boldsymbol{\beta}_0)\bigr]= \operatorname{E}\left\{ \frac{1}{n} \sum_{i=1}^n \frac{\delta_i}{\pi_i}\boldsymbol{X}_i^\top \boldsymbol{\Gamma}_i \boldsymbol{R}^{-1} \bigl[\tau-I(\boldsymbol{\varepsilon}_i \le 0)\bigr] \right\}= \boldsymbol{0}.
\]
Moreover, the variance of $\boldsymbol{S}_r(\boldsymbol{\beta}_0)$ can be decomposed as follows:
\[
\begin{aligned}
\operatorname{Var}\bigl[\boldsymbol{S}_r(\boldsymbol{\beta}_0)\bigr] 
&= \operatorname{E}\Biggl\{ \operatorname{Var}\Biggl[ \frac{1}{n} \sum_{i=1}^n \frac{\delta_i}{\pi_i} \boldsymbol{X}_i^\top \boldsymbol{\Gamma}_i \boldsymbol{R}^{-1} \bigl[\tau-I(\boldsymbol{\varepsilon}_i \le 0)\bigr] \Bigg|\mathcal{F}_n \Biggr] \Biggr\} \\
&\quad + \operatorname{Var}\Biggl\{\operatorname{E}\Biggl[ \frac{1}{n} \sum_{i=1}^n \frac{\delta_i}{\pi_i} \boldsymbol{X}_i^\top \boldsymbol{\Gamma}_i \boldsymbol{R}^{-1} \bigl[\tau-I(\boldsymbol{\varepsilon}_i \le 0)\bigr]\Bigg|\mathcal{F}_n \Biggr] \Biggr\}.
\end{aligned}
\]
The first term equals
\[
\frac{1}{n} \sum_{i=1}^{n} \frac{1-\pi_i}{n\pi_i} \boldsymbol{X}_i^\top \boldsymbol{\Gamma}_i \boldsymbol{R}^{-1}\boldsymbol{V}_0\boldsymbol{R}^{-1} \boldsymbol{\Gamma}_i \boldsymbol{X}_i,
\]
and the second term equals
\[
\frac{1}{n^2} \sum_{i=1}^{n} \boldsymbol{X}_i^\top \boldsymbol{\Gamma}_i \boldsymbol{R}^{-1}\boldsymbol{V}_0\boldsymbol{R}^{-1} \boldsymbol{\Gamma}_i \boldsymbol{X}_i.
\]
Therefore, combining these two terms yields
\[
\begin{aligned}
\operatorname{Var}\bigl[\boldsymbol{S}_r(\boldsymbol{\beta}_0)\bigr] 
&= \frac{1}{n} \sum_{i=1}^{n} \frac{1}{n\pi_i} \boldsymbol{X}_i^\top \boldsymbol{\Gamma}_i \boldsymbol{R}^{-1}\boldsymbol{V}_0\boldsymbol{R}^{-1} \boldsymbol{\Gamma}_i \boldsymbol{X}_i \\
&\leq \Biggl(\max_{1\leq i \leq n} \frac{1}{n\pi_i}\Biggr) \frac{1}{n} \sum_{i=1}^{n} \boldsymbol{X}_i^\top \boldsymbol{\Gamma}_i \boldsymbol{R}^{-1} \boldsymbol{V}_0\boldsymbol{R}^{-1} \boldsymbol{\Gamma}_i \boldsymbol{X}_i \\
&= O(1/r).
\end{aligned}
\]
The last equality follows from conditions (C4) and (C6).  Applying Chebyshev’s inequality yields the desired result.
\end{proof}

%%%%%%%%%%%%%%%%%%%%%%%%%%%%%%%%%%%%%%%%%%%%Lemma2
\begin{lemma}\label{lemma:Sr_SPhi}
Under conditions (C1)--(C6), if $n \to \infty$ and $r \to \infty$, then we have
 $$ \big\|\boldsymbol{S}_r\left(\boldsymbol{\beta}_{0}\right)-\boldsymbol{S}_r^{\Phi}\left(\boldsymbol{\beta}_{0}\right)\big\|=o_p(r^{-1/2}).$$
\end{lemma}

\begin{proof}[Proof of Lemma \ref{lemma:Sr_SPhi}] 

Results from Zu \textit{et al.} \cite{Zu23} imply that
$\big\|\operatorname{E}\big[\Phi(\boldsymbol{\varepsilon}_i/h) - I(\boldsymbol{\varepsilon}_i > 0)\big]\big\| = O(h^2)$ 
and 
$\operatorname{E}\big[\big\|\Phi(\boldsymbol{\varepsilon}_i/h) - I(\boldsymbol{\varepsilon}_i > 0)\big\|^2\big] = O(h).$
Define
$$\boldsymbol{S}_r\left(\boldsymbol{\beta}_{0}\right)-\boldsymbol{S}_r^{\Phi}\left(\boldsymbol{\beta}_{0}\right)=\frac{1}{n}\sum_{i=1}^{n} \frac{\delta_i}{\pi_i} \boldsymbol{X}_i^\top \boldsymbol{\Gamma}_i \boldsymbol{R}^{-1} \biggl[ 1-\Phi\left( \frac{\boldsymbol{\varepsilon}_i}{h} \right)- I(\boldsymbol{\varepsilon}_i \leq 0) \biggr]=\sum_{i=1}^{n} \boldsymbol{W}_i.$$
First, note that
\begin{align*}
\operatorname{E}\bigg[\bigg\|\sum_{i=1}^{n}\boldsymbol{W}_i\bigg\|^2\bigg] 
&= \operatorname{E}\bigg\{\bigg\|\sum_{i=1}^{n}\big[\operatorname{E}(\boldsymbol{W}_i) + (\boldsymbol{W}_i - \operatorname{E}(\boldsymbol{W}_i))\big]\bigg\|^2\bigg\} \\
&\leq 2\bigg\|\sum_{i=1}^{n}\operatorname{E}(\boldsymbol{W}_i)\bigg\|^2 + 2\operatorname{E}\bigg[\bigg\|\sum_{i=1}^{n}\big(\boldsymbol{W}_i - \operatorname{E}(\boldsymbol{W}_i)\big)\bigg\|^2\bigg] \\
&\leq 2\bigg\|\sum_{i=1}^{n}\operatorname{E}(\boldsymbol{W}_i)\bigg\|^2 + 2\sum_{i=1}^{n}\operatorname{E}\bigg(\bigg\|\boldsymbol{W}_i\bigg\|^2\bigg).
\end{align*}
Let $\boldsymbol{b}_i=1-\Phi\left( \boldsymbol{\varepsilon}_i/h\right)- I(\boldsymbol{\varepsilon}_i\leq 0)$. The two parts on the right side are treated separately. For the first part,
\begin{align*}
\bigg\|\sum_{i=1}^{n}\operatorname{E}(\boldsymbol{W}_i)\bigg\| 
& =  \bigg\|\frac{1}{n}\sum_{i=1}^{n}\operatorname{E}\bigg( \boldsymbol{X}_i^\top \boldsymbol{\Gamma}_i \boldsymbol{R}^{-1} \boldsymbol{b}_i\bigg) \bigg\|\\
& \leq \frac{1}{n}\sum_{i=1}^{n}\bigg[\bigg\|\operatorname{E}\bigg(\boldsymbol{X}_i^\top \boldsymbol{\Gamma}_i \boldsymbol{R}^{-1} \boldsymbol{b}_i\bigg)\bigg\|\bigg] \\
& \leq C\bigg[\max_{1\leq i \leq n}\big\|\boldsymbol{X}_i\big\|\bigg]\frac{1}{n}\sum_{i=1}^{n}\bigg[\big\|\operatorname{E}(\boldsymbol{b}_i)\big\|\bigg] \\
&= O(h^2),
\end{align*}
where C is a positive constant. Squaring both sides yields 
$\big\|\sum_{i=1}^{n}\operatorname{E}(\boldsymbol{W}_i)\big\|^2 =  O(h^4).$ 
We further derive that
\begin{align*}
\sum_{i=1}^{n}\operatorname{E}\bigg(\bigg\|\boldsymbol{W}_i\bigg\|^2\bigg)
&= \sum_{i=1}^{n}\operatorname{E}\bigg\{\bigg\|\frac{\delta_i}{n\pi_i} \boldsymbol{X}_i^\top \boldsymbol{\Gamma}_i \boldsymbol{R}^{-1}\boldsymbol{b}_i\bigg\|^2\bigg\} \\
&\leq \sum_{i=1}^{n}\frac{\|\boldsymbol{X}_i^\top \boldsymbol{\Gamma}_i \boldsymbol{R}^{-1}\|^2}{n} \operatorname{E}\bigg[\frac{1}{n\pi_i}\big\| \boldsymbol{b}_i\big\|^2\bigg] \\
&\leq C\Bigl(\max_{1\leq i \leq n} \frac{1}{n\pi_i}\Bigr)\frac{1}{n}\sum_{i=1}^{n}\operatorname{E}\left(\big\| \boldsymbol{b}_i\big\|^2\right) \\
&= O\left(h/r\right). \\
\end{align*}
As $h \asymp r^{-1/2}$, it can be deduced that
$$\big\|\boldsymbol{S}_r\left(\boldsymbol{\beta}_{0}\right)-\boldsymbol{S}_r^{\Phi}\left(\boldsymbol{\beta}_{0}\right)\big\|=o_p(r^{-1/2}).$$
\end{proof}

%%%%%%%%%%%%%%%%%%%%%%%%%%%%%%%%%%%%%%%%%%%%Lemma3
\begin{lemma}\label{lemma:Dr_H}
Under conditions (C1)--(C6), if $r \to \infty$, then we have
 $$ \sup_{\boldsymbol{\beta}\in\mathcal{B}_r}\left\|\boldsymbol{ D}_r(\boldsymbol\beta)-\boldsymbol H_n\right\|=o_p(1),$$
where $\mathcal{B}_r=\left\{\boldsymbol\beta:\|\boldsymbol\beta-\boldsymbol\beta_0\|\le a_r\right\},\quad a_r\to0$.
\end{lemma}

\begin{proof}[Proof of Lemma \ref{lemma:Dr_H}] 
By the triangle inequality,
\[
\begin{aligned}
&\sup_{\boldsymbol{\beta}\in\mathcal{B}_r}\left\|\boldsymbol{D}_r(\boldsymbol{\beta})-\boldsymbol{H}_n\right\|\\
&\leq \sup_{\boldsymbol{\beta}\in\mathcal{B}_r}\left\|\boldsymbol{D}_r(\boldsymbol{\beta})-E[\boldsymbol{D}_r(\boldsymbol{\beta})]\right\|+
\sup_{\boldsymbol{\beta}\in\mathcal{B}_r}
\left\|E[\boldsymbol{D}_r(\boldsymbol{\beta})]-\boldsymbol{H}_n\right\|.
\end{aligned}
\]
We first prove that
\[
\sup_{\boldsymbol{\beta}\in\mathcal{B}_r}\left\|\boldsymbol{D}_r(\boldsymbol{\beta})-E[\boldsymbol{D}_r(\boldsymbol{\beta})]\right\|=o_p(1).
\]
For $k_1,k_2=1,\cdots,p$, let $\boldsymbol{x}_{i,k}$ denote the $k$-th column of $\boldsymbol{X}_i$, and let $\varepsilon_{ij}(\boldsymbol{\beta})=y_{ij}-\boldsymbol{x}_{ij}^{\top}\boldsymbol{\beta}$. Define
\[
\begin{aligned}
q_i^{k_1,k_2}(\boldsymbol{\beta})&=\boldsymbol{x}_{i,k_1}^{\top}\boldsymbol{\Gamma}_i\boldsymbol{R}^{-1}\operatorname{diag}\left\{\frac1h\phi\left(\frac{\varepsilon_{ij}(\boldsymbol{\beta})}{h}\right)
\right\}_{j=1}^m\boldsymbol{x}_{i,k_2}, \\
Z_r^{k_1,k_2}(\boldsymbol{\beta})&=-\frac1n\sum_{i=1}^n\left\{\frac{\delta_i}{\pi_i}q_i^{k_1,k_2}(\boldsymbol{\beta})-E\left[\frac{\delta_i}{\pi_i}q_i^{k_1,k_2}(\boldsymbol{\beta})
\right]\right\}, \\
\xi_i(\boldsymbol{\beta})&=-\frac1n\left\{\frac{\delta_i}{\pi_i}q_i^{k_1,k_2}(\boldsymbol{\beta})-E\left[\frac{\delta_i}{\pi_i}q_i^{k_1,k_2}(\boldsymbol{\beta})\right]\right\},
\end{aligned}
\]
so that $Z_r^{k_1,k_2}(\boldsymbol{\beta})=\sum_{i=1}^n\xi_i(\boldsymbol{\beta})$.

Observations from different subjects are independent, so for each fixed $\boldsymbol{\beta}$, $\xi_1(\boldsymbol{\beta}),\cdots,\xi_n(\boldsymbol{\beta})$ are independent with $E[\xi_i(\boldsymbol{\beta})]=0$. Since $\sup_{\boldsymbol{\beta}\in\mathcal{B}_r}\left|q_i^{k_1,k_2}(\boldsymbol{\beta})\right|\leq \frac{C}{h}$,
\[
\begin{aligned}
|\xi_i(\boldsymbol{\beta})|
&\leq \frac1n\frac{\delta_i}{\pi_i}\left|q_i^{k_1,k_2}(\boldsymbol{\beta})\right| + \frac1n E\left[\left|q_i^{k_1,k_2}(\boldsymbol{\beta})\right|\right] \\
& \leq \Biggl(\max_{1\leq i \leq n} \frac{1}{n\pi_i}\Biggr)\frac{C}{h}+ \frac{C}{nh}\\
& \leq  \frac{C}{rh}
\end{aligned}
\]
Set $b_r=C/(rh)$, so that $|\xi_i(\boldsymbol{\beta})|\leq b_r$ uniformly in $\boldsymbol{\beta}\in\mathcal{B}_r$. Moreover,
\[
\begin{aligned}
\operatorname{Var}\{\xi_i(\boldsymbol{\beta})\}
&\leq\frac1{n^2}E\left\{\left[\frac{\delta_i}{\pi_i}q_i^{k_1,k_2}(\boldsymbol{\beta})\right]^2\right\}\\
&\leq\Biggl(\max_{1\leq i \leq n} \frac{1}{n\pi_i}\Biggr)\frac1{n}E\left\{\left[q_i^{k_1,k_2}(\boldsymbol{\beta})\right]^2\right\},
\end{aligned}
\]
We now show that $\sup_{\boldsymbol{\beta}\in\mathcal{B}_r}E\left\{\left[q_i^{k_1,k_2}(\boldsymbol{\beta})\right]^2\right\}=O\left(\frac1h\right)$. Note that $\left|q_i^{k_1,k_2}(\boldsymbol{\beta})\right|^2\le C\sum_{j=1}^m\frac1{h^2}\phi^2\left(\frac{\varepsilon_{ij}(\boldsymbol{\beta})}{h}\right).$ Let $d_{ij}(\boldsymbol{\beta})=\boldsymbol{x}_{ij}^{\top}(\boldsymbol{\beta}-\boldsymbol{\beta}_0)$. Then $\varepsilon_{ij}(\boldsymbol{\beta})=\varepsilon_{ij}-d_{ij}(\boldsymbol{\beta})$. Thus
\[
\begin{aligned}
E\left[\frac1{h^2}\phi^2\left(\frac{\varepsilon_{ij}(\boldsymbol{\beta})}{h}\right)\right]
&=\frac1{h^2}\int\phi^2\left(\frac{u-d_{ij}(\boldsymbol{\beta})}{h}\right)f_{ij}(u)\,du\\
&=\frac1h\int\phi^2(v)f_{ij}\{d_{ij}(\boldsymbol{\beta})+hv\}\,dv\\
&\leq \frac{C}{h}.
\end{aligned}
\]
Consequently, $\sup_{\boldsymbol{\beta}\in\mathcal{B}_r}E\left\{\left[q_i^{k_1,k_2}(\boldsymbol{\beta})\right]^2\right\}\le \frac{C}{h}$. This yields $\operatorname{Var}[\xi_i(\boldsymbol{\beta})]\le C/(nrh)$, and hence $\sum_{i=1}^n\operatorname{Var}\{\xi_i(\boldsymbol{\beta})\}\leq\frac{C}{rh}$. We set $V_r=C/(rh)$ as an upper bound for the total variance.

The argument is similar to that used in the proof of Lemma B.5 in \cite{Sherwood16}. We take the grid resolution to be $\rho_r=h^3$. There exist finitely many points $\boldsymbol{\beta}_1,\cdots,\boldsymbol{\beta}_{N_r}$ such that, for any $\boldsymbol{\beta}\in\mathcal{B}_r$, there exists some $\boldsymbol{\beta}_l$ satisfying $\|\boldsymbol{\beta}-\boldsymbol{\beta}_l\|\leq h^3$. The covering number obeys $N_r \leq \left(1+\frac{2a_r}{h^3}\right)^p$. Therefore, $\log N_r = O\bigl(\log(1/h)\bigr)$. For any fixed grid point $\boldsymbol{\beta}_l$, Bernstein's inequality gives
\[
P\left[|Z_r^{k_1,k_2}(\boldsymbol{\beta}_l)|>\sqrt{2V_rx}+\frac{b_rx}{3}\right]\le 2e^{-x}.
\]
Applying the union bound over all $N_r$ grid points:
\[
\begin{aligned}
P\left[\max_{1\le l\le N_r}|Z_r^{k_1,k_2}(\boldsymbol{\beta}_l)|>\sqrt{2V_rx}+\frac{b_rx}{3}\right]\leq 2N_re^{-x}.
\end{aligned}
\]
Taking $x=M\log(2N_r)$ with $M>1$, we obtain $2N_re^{-x}=(2N_r)^{1-M}\to0$. Therefore,
\[
\max_{1\le l\le N_r}|Z_r^{k_1,k_2}(\boldsymbol{\beta}_l)|=O_p\left(\sqrt{V_r\log N_r}+b_r\log N_r\right).
\]
We next bound the fluctuations between each $\boldsymbol{\beta}\in\mathcal{B}_r$ and its corresponding grid point. Since $\varepsilon_{ij}(\boldsymbol{\beta})-\varepsilon_{ij}(\boldsymbol{\beta}_l)=-\boldsymbol{x}_{ij}^{\top}(\boldsymbol{\beta}-\boldsymbol{\beta}_l)$, we have
\[
\left|q_i^{k_1,k_2}(\boldsymbol{\beta})-q_i^{k_1,k_2}(\boldsymbol{\beta}_l)\right|\le\frac C{h^2}\|\boldsymbol{\beta}-\boldsymbol{\beta}_l\|\le h.
\]
Therefore,
\[
\begin{aligned}
\left|Z_r^{k_1,k_2}(\boldsymbol{\beta})-Z_r^{k_1,k_2}(\boldsymbol{\beta}_l)\right|
&\le Ch\left[\frac1n\sum_{i=1}^n\frac{\delta_i}{\pi_i}+1\right].
\end{aligned}
\]
By Markov's inequality, $\frac1n\sum_{i=1}^n\frac{\delta_i}{\pi_i}=O_p(1)$. Hence,
\[
\sup_{\boldsymbol{\beta}\in\mathcal{B}_r}\left|Z_r^{k_1,k_2}(\boldsymbol{\beta})-Z_r^{k_1,k_2}(\boldsymbol{\beta}_l)\right|=O_p(h).
\]
Combining the above two bounds,
\[
\begin{aligned}
\sup_{\boldsymbol{\beta}\in\mathcal{B}_r}\left|Z_r^{k_1,k_2}(\boldsymbol{\beta})\right|
&\leq \max_{1\le l\le N_r}\left|Z_r^{k_1,k_2}(\boldsymbol{\beta}_l)\right|
+\sup_{\boldsymbol{\beta}\in\mathcal{B}_r}\left|Z_r^{k_1,k_2}(\boldsymbol{\beta})-Z_r^{k_1,k_2}\big(\boldsymbol{\beta}_l\big)\right| \\
&= O_p\left(\sqrt{\frac{\log(1/h)}{rh}}+\frac{\log(1/h)}{rh}+h\right) = o_p(1).
\end{aligned}
\]
It then follows that $\sup_{\boldsymbol{\beta}\in\mathcal{B}_r}\left\|\boldsymbol{D}_r(\boldsymbol{\beta})-E[\boldsymbol{D}_r(\boldsymbol{\beta})]\right\|=o_p(1)$.

We now prove $\sup_{\boldsymbol{\beta}\in\mathcal{B}_r}\left\|E\left[\boldsymbol{D}_r(\boldsymbol{\beta})\right]-\boldsymbol{H}_n\right\|=o(1)$. The following argument is similar to that used in the proof of Theorem 1 in \cite{Fu12}. Define
\[
E\left[\boldsymbol{D}_r(\boldsymbol{\beta})\right]=-\frac1n\sum_{i=1}^n\boldsymbol{X}_i^\top\boldsymbol{\Gamma}_i\boldsymbol{R}^{-1}\operatorname{diag}\left
\{g_{ij}(\boldsymbol{\beta})\right\}_{j=1}^m\boldsymbol{X}_i,
\]
where $g_{ij}(\boldsymbol{\beta})=E\left[\frac1h\phi\left(\frac{\varepsilon_{ij}(\boldsymbol{\beta})}{h}\right)\right]$. Then
\[
\begin{aligned}
g_{ij}(\boldsymbol{\beta})-f_{ij}(0)
&=\int\frac1h\phi\left(\frac{u-d_{ij}(\boldsymbol{\beta})}{h}\right)f_{ij}(u)du-f_{ij}(0)\\
&=\int\phi(v)f_{ij}\left\{d_{ij}(\boldsymbol{\beta})+hv\right\}dv-f_{ij}(0)\\
&=\int\phi(v)\left[f_{ij}\left\{d_{ij}(\boldsymbol{\beta})+hv\right\}-f_{ij}(0)\right]dv.
\end{aligned}
\]
The mean value theorem gives $\left|f_{ij}\left\{d_{ij}(\boldsymbol{\beta})+hv\right\}-f_{ij}(0)\right|\leq C(a_r+h|v|)$. Thus
$\sup_{\boldsymbol{\beta}\in\mathcal{B}_r}|g_{ij}(\boldsymbol{\beta})-f_{ij}(0)|=o(1)$. Recalling that $\boldsymbol{\Gamma}_i=\operatorname{diag}\{f_{ij}(0)\}_{j=1}^m$, we have
\[
\begin{aligned}
E\left[\boldsymbol{D}_r(\boldsymbol{\beta})\right]-\boldsymbol{H}_n&=-\frac1n\sum_{i=1}^n\boldsymbol{X}_i^\top\boldsymbol{\Gamma}_i\boldsymbol{R}^{-1}\big[\operatorname{diag}\{g_{ij}(\boldsymbol{\beta})\}_{j=1}^m-\boldsymbol{\Gamma}_i\big]\boldsymbol{X}_i.
\end{aligned}
\]
It follows that
\[
\begin{aligned}
\sup_{\boldsymbol{\beta}\in\mathcal{B}_r}\left\|E\left[\boldsymbol{D}_r(\boldsymbol{\beta})\right]-\boldsymbol{H}_n\right\|&\leq C\sup_{\boldsymbol{\beta}\in\mathcal{B}_r}|g_{ij}(\boldsymbol{\beta})-f_{ij}(0)|\\
&=o(1).
\end{aligned}
\]
Therefore, the proof is completed.
\end{proof}

%%%%%%%%%%%%%%%%%%%%%%%%%%%%%%%%%%%%%%%%%%%%Theorem 1
\begin{proof}[\noindent\textbf{\emph{The proof of Theorem 2.1 }}]
Define
\begin{align*}
\boldsymbol{M}_n(\boldsymbol{\beta}) &= \operatorname{E}\left\{\frac{1}{n}\sum_{i=1}^n\boldsymbol{X}_i^\top\boldsymbol{\Gamma}_i\boldsymbol{R}^{-1}\big[\tau-I(\boldsymbol{Y}_i\leq\boldsymbol{X}_i\boldsymbol{\beta})\big]\right\} \\
&= \frac{1}{n}\sum_{i=1}^n\boldsymbol{X}_i^\top\boldsymbol{\Gamma}_i\boldsymbol{R}^{-1}\bigg[\operatorname{P}(\boldsymbol{Y}_i>\boldsymbol{X}_i\boldsymbol{\beta})-(1-\tau)\bigg].
\end{align*}
At the true parameter $\boldsymbol{\beta}_0$,
$\operatorname{P}(y_{ij}>\boldsymbol{x}_{ij}^\top\boldsymbol{\beta}_0)=\operatorname{P}(\varepsilon_{ij}>0)=1-\tau$, which implies $\boldsymbol{M}_n(\boldsymbol{\beta}_0)=\boldsymbol{0}$. We first show that $\sup_{\boldsymbol{\beta}\in\Lambda}\big\|\boldsymbol{S}_r^{\Phi}(\boldsymbol{\beta})-\boldsymbol{M}_n(\boldsymbol{\beta})\big\|=o_p(1)$. By the triangle inequality,
\begin{align*}
\sup_{\boldsymbol{\beta}\in\Lambda}\big\|\boldsymbol{S}_r^{\Phi}(\boldsymbol{\beta})-\boldsymbol{M}_n(\boldsymbol{\beta})\big\|
&\leq \sup_{\boldsymbol{\beta}\in\Lambda}\Big\|\boldsymbol{S}_r^{\Phi}(\boldsymbol{\beta})-\operatorname{E}\big[\boldsymbol{S}_r^{\Phi}(\boldsymbol{\beta})\big]\Big\| \\
&\quad + \sup_{\boldsymbol{\beta}\in\Lambda}\Big\|\operatorname{E}\big[\boldsymbol{S}_r^{\Phi}(\boldsymbol{\beta})\big]-\boldsymbol{M}_n(\boldsymbol{\beta})\Big\|.
\end{align*}
Similar to Lemma \ref{lemma:Dr_H}, set the grid resolution $\rho_r=h^2$. We can directly prove that
\[
\sup_{\boldsymbol{\beta}\in\Lambda}\Big\|\boldsymbol{S}_r^{\Phi}(\boldsymbol{\beta})-\operatorname{E}\big[\boldsymbol{S}_r^{\Phi}(\boldsymbol{\beta})\big]\Big\|=o_p(1).
\]
Note that
\begin{align*}
& \operatorname{E}\left[\Phi\left(\frac{\varepsilon_{ij}-d_{ij}(\boldsymbol{\beta})}{h}\right)\right]-\operatorname{P}\big[\varepsilon_{ij}>d_{ij}(\boldsymbol{\beta})\big] \\
=& \int_{\mathbb{R}}\left[\Phi\left(\frac{u-d_{ij}(\boldsymbol{\beta})}{h}\right)-I[u>d_{ij}(\boldsymbol{\beta})]\right]f_{ij}(u)\,\mathrm{d}u \\
=& h\int_{\mathbb{R}}\big[\Phi(v)-I(v>0)\big]f_{ij}\big(d_{ij}(\boldsymbol{\beta})+hv\big)\,\mathrm{d}v\\
\leq& Ch\int_{\mathbb{R}}\big|\Phi(v)-I(v>0)\big|\mathrm{d}v\\
\leq& Ch,
\end{align*}
where $d_{ij}(\boldsymbol{\beta})=\boldsymbol{x}_{ij}^\top(\boldsymbol{\beta}-\boldsymbol{\beta}_0)$.

Substituting the uniform bound derived above, we obtain
\[
\begin{aligned}
\sup_{\boldsymbol{\beta}\in\Lambda}\Big\|\operatorname{E}\big[\boldsymbol{S}_r^{\Phi}(\boldsymbol{\beta})\big]-\boldsymbol{M}_n(\boldsymbol{\beta})\Big\|
&\leq \sup_{\boldsymbol{\beta}\in\Lambda}\frac{1}{n}\sum_{i=1}^{n}\big\|\boldsymbol{X}_i^\top\boldsymbol{\Gamma}_i\boldsymbol{R}^{-1}\big\| \cdot \left\|\operatorname{E}\left[\Phi\left(\frac{\boldsymbol{\varepsilon}_i(\boldsymbol{\beta})}{h}\right)\right]-\operatorname{P}\big[\boldsymbol{\varepsilon}_i(\boldsymbol{\beta})>\boldsymbol{0}\big]\right\| \\
&\leq Ch,
\end{aligned}
\]
which implies $\sup_{\boldsymbol{\beta}\in\Lambda}\Big\|\operatorname{E}\big[\boldsymbol{S}_r^{\Phi}(\boldsymbol{\beta})\big]-\boldsymbol{M}_n(\boldsymbol{\beta})\Big\| = o_p(1)$. Thus, we have shown that
\[
\sup_{\boldsymbol{\beta}\in\Lambda}\big\|\boldsymbol{S}_r^{\Phi}(\boldsymbol{\beta})-\boldsymbol{M}_n(\boldsymbol{\beta})\big\|=o_p(1).
\]
Note that $\boldsymbol{\beta}_0$ is the unique root of $\boldsymbol{M}_n(\boldsymbol{\beta})=\boldsymbol{0}$, while $\tilde{\boldsymbol{\beta}}$ satisfies $\boldsymbol{S}_r^{\Phi}(\tilde{\boldsymbol{\beta}})=\boldsymbol{0}$. It then follows that $\|\tilde{\boldsymbol{\beta}}-\boldsymbol{\beta}_0\|=o_p(1)$.

A first-order Taylor expansion of $\boldsymbol{S}_r^{\Phi}(\boldsymbol{\beta})$ at $\boldsymbol{\beta}_0$ yields
\[
\boldsymbol{S}_r^{\Phi}(\boldsymbol{\beta})=\boldsymbol{S}_r^{\Phi}(\boldsymbol{\beta}_0)+\boldsymbol{D}_r(\boldsymbol{\beta}^*)(\boldsymbol{\beta}-\boldsymbol{\beta}_0),
\]
where $\boldsymbol{\beta}^*$ lies between $\boldsymbol{\beta}$ and $\boldsymbol{\beta}_0$, and \[
\boldsymbol{D}_r(\boldsymbol{\beta}^*)=-\frac{1}{n}\sum_{i=1}^{n}\frac{\delta_i}{\pi_i}\boldsymbol{X}_i^\top \boldsymbol{\Gamma}_i \boldsymbol{R}^{-1}\operatorname{diag}\left\{\frac{1}{h}\phi\left(\frac{\varepsilon_{ij}(\boldsymbol{\beta}^*)}{h}\right)\right\}_{j=1}^m\boldsymbol{X}_i.
\]
By Lemma \ref{lemma:Dr_H}, we have $\boldsymbol{D}_r(\boldsymbol{\beta}^*) - \boldsymbol{H}_n = o_p(1)$. Setting $\boldsymbol{\beta}=\tilde{\boldsymbol{\beta}}$, we obtain
\[
\boldsymbol{0}=\boldsymbol{S}_r^{\Phi}(\boldsymbol{\beta}_0)+\boldsymbol{H}_n(\tilde{\boldsymbol{\beta}}-\boldsymbol{\beta}_0)+o_p(\|\tilde{\boldsymbol{\beta}}-\boldsymbol{\beta}_0\|).
\]
By Lemma \ref{lemma:Sr} and Lemma \ref{lemma:Sr_SPhi},
\[
\tilde{\boldsymbol{\beta}}-\boldsymbol{\beta}_0
=-\boldsymbol{H}_n^{-1}\boldsymbol{S}_r(\boldsymbol{\beta}_0)+o_p(\|\tilde{\boldsymbol{\beta}}-\boldsymbol{\beta}_0\|)+o_p(r^{-1/2})
=O_p(r^{-1/2}).
\]
\end{proof}

%%%%%%%%%%%%%%%%%%%%%%%%%%%%%%%%%%%%%%%%%%%%Theorem 2
\begin{proof}[\noindent\textbf{\emph{The proof of Theorem 2.2}}]

Let $\boldsymbol{S}_r(\boldsymbol{\beta}_0)=\frac{1}{n}\sum_{i=1}^{n} \frac{\delta_i}{\pi_i} \boldsymbol{X}_i^\top \boldsymbol{\Gamma}_i \boldsymbol{R}^{-1} \big[\tau-I(\boldsymbol{\varepsilon}_i\leq0)\big]=\sum_{i=1}^{n}\boldsymbol{T}_i$. The conditions of the Lindeberg-Feller central limit theorem are verified as follows. For any $\varepsilon>0$,

$$\begin{aligned}
&\sum_{i=1}^n \operatorname{E}\big[\left\|\boldsymbol{T}_i\right\|^{2} I(\|\boldsymbol{T}_i\| \geq \varepsilon)\big] 
 \leq \frac{1}{\varepsilon}\sum_{i=1}^{n} \operatorname{E}\left(\left\|\boldsymbol{T}_i\right\|^3\right) \\
&= \frac{1}{\varepsilon}\sum_{i=1}^{n}  \operatorname{E}\bigg[\bigg\| \frac{\delta_i}{n\pi_i} \boldsymbol{X}_i^\top \boldsymbol{\Gamma}_i \boldsymbol{R}^{-1} \big( \tau - I(\boldsymbol{\varepsilon}_i \leq 0) \big) \bigg\|^3\bigg] \\
&=O\left(1/r^2\right)
\end{aligned}$$
Thus, the Lindeberg condition is satisfied. By Theorem 2.1, we further derive
\[
\boldsymbol{V}^{-1/2}(\tilde{\boldsymbol{\beta}} - \boldsymbol{\beta}_0)
= -\boldsymbol{V}^{-1/2}\boldsymbol{H}_n^{-1}\boldsymbol{V}_c^{1/2}\boldsymbol{V}_c^{-1/2}\boldsymbol{S}_r(\boldsymbol{\beta}_0)+o_p(1),
\]
where $\boldsymbol{V} = \boldsymbol{H}_n^{-1} \boldsymbol{V}_c \boldsymbol{H}_n^{-1}$. Via Slutsky’s theorem, the asymptotic normality result follows,
\[
\boldsymbol{V}^{-1/2}(\tilde{\boldsymbol{\beta}}-\boldsymbol{\beta}_0)\xrightarrow{d} N(\boldsymbol{0}, \boldsymbol{I}).
\]
\end{proof}

%%%%%%%%%%%%%%%%%%%%%%%%%%%%%%%%%%Theorem 3
\begin{proof}[\noindent\textbf{\emph{The proof of Theorem 3.1}}]
If some elements of $\left\{z_i\right\}_{i=1}^n$ equal zero, their associated subsampling probabilities are set to zero, and the subsampling probabilities of the remaining individuals are considered. Thus, all \( z_i > 0 \) are assumed without loss of generality. 

To minimize $\mathrm{tr}\left( \boldsymbol{V}\right)$, which corresponds to the asymptotic mean squared error, 
\begin{equation}
\begin{aligned} \label{eq:opt_problem}
& \min G = \min \sum_{i=1}^n \mathrm{tr} \left(\frac{1}{n^2\pi_i} \boldsymbol{H}_n^{-1} \boldsymbol{X}_i^\top\boldsymbol{\Gamma}_i\boldsymbol{R}^{-1} \boldsymbol{V}_0\boldsymbol{R}^{-1}\boldsymbol{\Gamma}_i\boldsymbol{X}_i \boldsymbol{H}_n^{-1}\right) \\ 
& \operatorname{s.t.} \quad \sum_{i=1}^n \pi_i = r, \quad 0 \leq \pi_i \leq 1, \quad i = 1, \cdots, n.
\end{aligned}
\end{equation}

We assume an ordered sequence \( z_1 \leq z_2 \leq \dots \leq z_n \), which does not restrict generality. Applying the Cauchy-Schwarz inequality,  
\begin{equation}
\label{eq:G}
\begin{aligned}
G 
&= \sum_{i=1}^n \mathrm{tr}\left(
    \frac{1}{n^2\pi_i} \boldsymbol{H}_n^{-1} \boldsymbol{X}_i^\top\boldsymbol{\Gamma}_i\boldsymbol{R}^{-1} 
    \boldsymbol{V}_0\boldsymbol{R}^{-1}\boldsymbol{\Gamma}_i\boldsymbol{X}_i \boldsymbol{H}_n^{-1}
\right) \\  
&= \frac{\tau(1-\tau)}{n^2}\sum_{i=1}^{n}\frac{1}{\pi_i}\left\|
    \boldsymbol{H}_n^{-1} \boldsymbol{X}_i^\top\boldsymbol{\Gamma}_i\boldsymbol{R}^{-1}\boldsymbol{V}_0^{1/2} 
\right\|^2 \\
&= \frac{\tau(1-\tau)}{n^2} \sum_{i=1}^n \left(\pi_i^{-1} z_i^2 \right) \\  
&= \frac{\tau(1-\tau)}{n^2r}\left(\sum_{i=1}^n \pi_i\right)\left(\sum_{i=1}^n \pi_i^{-1} z_i^2 \right) \\  
&\geq \frac{\tau(1-\tau)}{n^2r} \left(\sum_{i=1}^n z_i \right)^2,
\end{aligned}
\end{equation}
with equality holding if and only if $\pi_i \propto z_i$. Therefore, it can be stated that $\mathrm{tr}\left( \boldsymbol{V}\right)$ attains its minimum value when $\pi_{i} \propto z_{i}$.

Two cases are considered as follows: (1) If $r z_i / \sum_{j=1}^n z_j \leq 1$ for all $i=1,\cdots,n$, then $\pi_i = r z_i / \sum_{j=1}^n z_j$; (2) If there exist some $i$ such that $r z_i / \sum_{j=1}^n z_j > 1$, then there must be exactly $u$ such indices $i$. Specifically, according to the definition of $u$,
$ rz_i > \sum_{j=1}^n z_j = \sum_{j=1}^{n-u}z_j+\sum_{j=n-u+1}^{n}z_j > (r-u)z_{n-u}+uz_{n-u}=rz_{n-u}$
which implies $i > n-u$. In this case, (\ref{eq:opt_problem}) is equivalent to
\begin{equation}
\begin{aligned}\label{eq:opt_problem2}
\min G &= \min \sum_{i=1}^{n-u} \operatorname{tr}\left(
\frac{1}{n^2 \pi_i} \boldsymbol{H}_n^{-1} \boldsymbol{X}_i^\top \boldsymbol{\Gamma}_i \boldsymbol{R}^{-1} 
\boldsymbol{V}_0\boldsymbol{R}^{-1} \boldsymbol{\Gamma}_i \boldsymbol{X}_i \boldsymbol{H}_n^{-1}
\right) \\
\text{s.t.} \quad &\sum_{i=1}^{n-u} \pi_i = r-u, \quad 0 \leq \pi_i \leq 1, \quad i = 1, \cdots, n-u, \quad\pi_{n-u+1} = \cdots = \pi_n=1.
\end{aligned}
\end{equation}

Similarly, an application of the Cauchy–Schwarz inequality yields
\[
\begin{aligned}
&\frac{\tau(1-\tau)}{n^2}\sum_{i=1}^{n-u}\frac{1}{\pi_i}\left\|\boldsymbol{H}_n^{-1} \boldsymbol{X}_i^\top\boldsymbol{\Gamma}_i\boldsymbol{R}^{-1}\boldsymbol{V}_0^{1/2}\right\|^2 \\
&= \frac{\tau(1-\tau)}{n^2(r-u)}\left(\sum_{i=1}^{n-u} \pi_i\right)\left(\sum_{i=1}^{n-u} \pi_i^{-1} z_i^2 \right) \\
&\geq \frac{\tau(1-\tau)}{n^2(r-u)} \left(\sum_{i=1}^{n-u} z_i \right)^2.
\end{aligned}
\]
The trace $\mathrm{tr}\left( \boldsymbol{V}\right)$ attains its minimum value when
\[
\pi_i = 
\begin{cases}
\displaystyle (r-u)z_i/\sum_{j=1}^{n-u} z_j, & i = 1, \cdots, n-u, \\
1, & i = n-u+1, \cdots, n.
\end{cases}
\]

Next, we aim to unify the expressions for $\pi_i$. Suppose there exists a constant $T$ such that
\[\max_{1 \leq i \leq n} \frac{z_i \wedge T}{\sum_{j=1}^n (z_j \wedge T)} = \frac{1}{r},\]
and $z_{n-u} < T \leq z_{n-u+1}$. This implies that $\sum_{i=1}^{n-u} z_i = (r-u)T$. By (\ref{eq:G}), we obtain
\[
\begin{aligned}
\min G
&= \frac{\tau(1-\tau)}{n^2} \sum_{i=1}^n \left(\pi_i^{-1} z_i^2 \right) \\ 
&= \frac{\tau(1-\tau)}{n^2}\sum_{i=1}^{n-u}\left(\pi_i^{-1} z_i^2 \right)+\frac{\tau(1-\tau)}{n^2}\sum_{i=n-u+1}^{n}\left(\pi_i^{-1} z_i^2 \right)\\  
&= \frac{\tau(1-\tau)}{n^2(r-u)}\left(\sum_{i=1}^{n-u}z_i\right)^2+\frac{\tau(1-\tau)}{n^2}\sum_{i=n-u+1}^{n}z_i^2\\  
&= \frac{\tau(1-\tau)T^2 (r-u)}{n^2} + \frac{\tau(1-\tau)}{n^2}\sum_{i=n-u+1}^{n}z_i^2. 
\end{aligned}
\]
          
Let $\pi_i = r (z_i \wedge T)/\sum_{j=1}^n (z_j \wedge T)$. Substituting it into (\ref{eq:opt_problem2}), we obtain
\begin{equation*}
\begin{aligned}
G 
&= \frac{\tau(1-\tau)}{n^2} \sum_{i=1}^n \left(\pi_i^{-1} z_i^2 \right) \\  
&= \frac{\tau(1-\tau)}{n^2} \sum_{i=1}^{n-u} \left(\pi_i^{-1} z_i^2 \right)+\frac{\tau(1-\tau)}{n^2} \sum_{i=n-u+1}^{n} \left(\pi_i^{-1} z_i^2 \right)\\
&= \frac{\tau(1-\tau)}{n^2r} \sum_{i=1}^{n-u}\frac{\sum_{j=1}^{n}(z_j\wedge T)}{z_i\wedge T} z_i^2+\frac{\tau(1-\tau)}{n^2} \sum_{i=n-u+1}^{n} z_i^2 \\
&= \frac{\tau(1-\tau)}{n^2r} \sum_{i=1}^{n-u}\frac{\sum_{j=1}^{n-u}z_j+uT}{z_i} z_i^2+\frac{\tau(1-\tau)}{n^2} \sum_{i=n-u+1}^{n} z_i^2\\
&= \frac{\tau(1-\tau)T^2(r-u)}{n^2} +\frac{\tau(1-\tau)}{n^2} \sum_{i=n-u+1}^{n} z_i^2\\
&=\min G,
\end{aligned}
\end{equation*}
which implies $\pi_i$ is the optimal solution of (\ref{eq:opt_problem2}).
 
We next examine the existence of $T$, together with the condition $z_{n-u} < T \leq z_{n-u+1}$. By the definition of $u$, we have that
\[\frac{(r-u+1)z_{n-u+1}}{\sum_{i=1}^{n-u+1}z_i} \geq 1 \quad \text{and} \quad \frac{(r-u)z_{n-u}}{\sum_{i=1}^{n-u}z_i} < 1.\]
Let $T_1 = z_{n-u+1}$ and $T_2 = z_{n-u}$. The following inequalities hold:
\[\frac{(r-u+1)z_{n-u+1} + (u-1)T_1}{\sum_{i=1}^{n-u+1}z_i + (u-1)T_1} \geq 1,\]
and
\[\frac{(r-u)z_{n-u} + uT_2}{\sum_{i=1}^{n-u}z_i + uT_2} < 1.\]
This yields the results
\[\frac{z_i \wedge T_1}{\sum_{j=1}^{n} z_j \wedge T_1} \geq \frac{1}{r} \quad \text{and} \quad \frac{z_i \wedge T_2}{\sum_{j=1}^{n} z_j \wedge T_2} < \frac{1}{r}.\]
Note that the function $\max\limits_{1\leq i\leq n} \dfrac{r(z_i\wedge T)}{\sum_{j=1}^n(z_j\wedge T)} = 1$
is continuous with respect to $T$. Thus, there must exist a $T$ satisfying $z_{n-u} < T \leq z_{n-u+1}$ such that
$\max\limits_{1\leq i\leq n} \dfrac{r(z_i\wedge T)}{\sum_{j=1}^n(z_j\wedge T)} = 1$.
Case (1) is a special case of Case (2) at $u=0$.
\end{proof}

%%%%%%%%%%%%%%%%%%%%%%%%%%%%%%%%%%Theorem 4
\begin{proof}[Proof of Theorem 4.1]
To prove conclusion (a), we construct the estimator $\check{\boldsymbol{\beta}} = \big(\check{\boldsymbol{\beta}}_1^\top, \boldsymbol{0}^\top\big)^\top$, with $\check{\boldsymbol{\beta}}_1 = \boldsymbol{\beta}_{01} - \boldsymbol{H}_{n11}^{-1}\boldsymbol{S}_{r1}^\Phi(\boldsymbol{\beta}_0)$.
Here, $\boldsymbol{H}_{n11}$ denotes the leading $s\times s$ submatrix of $\boldsymbol{H}_{n}$, and $\boldsymbol{\beta}_{01} = (\beta_{01}, \cdots, \beta_{0s})^\top$ is the vector of the first $s$ components of the true parameter $\boldsymbol{\beta}_0$.
For $k=1,\cdots,s$, under condition (C7)(i) and $\check{\boldsymbol{\beta}} = \boldsymbol{\beta}_0 + O_p(r^{-1/2})$, as $\epsilon\to0^+$, we obtain
\[
\begin{aligned}
r^{1/2} S_{rk}^P\big(\check{\boldsymbol{\beta}} \pm \epsilon \boldsymbol{e}_k\big)
&= r^{1/2} S_{rk}^{\Phi}\big(\check{\boldsymbol{\beta}} \pm \epsilon \boldsymbol{e}_k\big) - r^{1/2} \lambda \tilde{w}_k\text{sgn}\big(\check{\beta}_k \pm \epsilon\big) \\
&= o_p(1).
\end{aligned}
\]
For $k=s+1,\cdots,p$, under condition (C7)(ii) and $\epsilon\to0^+$, $r^{1/2} S_{rk}^P\big(\check{\boldsymbol{\beta}} +\epsilon \boldsymbol{e}_k\big)$ and $r^{1/2} S_{rk}^P\big(\check{\boldsymbol{\beta}} -\epsilon \boldsymbol{e}_k\big)$ are dominated by $-r^{1/2}\lambda \tilde{w}_k$ and $r^{1/2}\lambda \tilde{w}_k$, which have opposite signs.
Consequently, $\check{\boldsymbol{\beta}}$ is an approximate zero-crossing point of $\boldsymbol{S}_{r}^P(\boldsymbol{\beta})$.

To prove conclusion (b), consider the sets $C_k=\{\check{\beta}_k\neq0\}$, $k=s+1,\cdots,p$, in the probability space. It suffices to show that for any $\epsilon>0$, the inequality $\operatorname{P}(C_k)<\epsilon$ holds when $r$ is sufficiently large. Since $\check{\beta}_k=O_p(r^{-1/2})$, there exists some $M$ such that, for sufficiently large $r$,
\[
\begin{aligned}
\text{P}(C_k)
&= \operatorname{P}(\check{\beta}_k \neq 0) \\
&= \operatorname{P}\Big(\big\{\check{\beta}_k \neq 0,\, |\check{\beta}_k| \geq M r^{-1/2}\big\}
\cup \big\{\check{\beta}_k \neq 0,\, |\check{\beta}_k| < M r^{-1/2}\big\}\Big) \\
&\leq \operatorname{P}\big(\check{\beta}_k \neq 0,\, |\check{\beta}_k| \geq M r^{-1/2}\big)
+ \operatorname{P}\big(\check{\beta}_k \neq 0,\, |\check{\beta}_k| < M r^{-1/2}\big)\\
&\leq \frac{\epsilon}{2}
+ \operatorname{P}\big(\check{\beta}_k \neq 0,\, |\check{\beta}_k| < M r^{-1/2}\big).
\end{aligned}
\]
Subsequently, we focus on $\operatorname{P}\big(\check{\beta}_k \neq 0,\, |\check{\beta}_k| < M r^{-1/2}\big)$. By virtue of the $k$-th component of the penalized estimating functions and the definition of approximate zero-crossing point,
\[
\begin{aligned}
r^{1/2}S_{rk}^P(\check{\boldsymbol{\beta}})
&= r^{1/2}S_{rk}^\Phi(\boldsymbol{\beta}_0) + r^{1/2}\boldsymbol{H}_{nk}(\check{\boldsymbol{\beta}}-\boldsymbol{\beta}_0) + o_p(1) - r^{1/2}\lambda \tilde{w}_k\operatorname{sgn}(\check{\beta}_k) \\
&= o_p(1).
\end{aligned}
\]
Accordingly,
\[
o_p(1) = \left\{r^{1/2}S_{rk}^\Phi(\boldsymbol{\beta}_0) + r^{1/2}\boldsymbol{H}_{nk}(\check{\boldsymbol{\beta}}-\boldsymbol{\beta}_0) + o_p(1) - r^{1/2}\lambda \tilde{w}_k\operatorname{sgn}(\check{\beta}_k)\right\}^2,
\]
where $\boldsymbol{H}_{nk}$ denotes the $k$-th row of $\boldsymbol{H}_n$.
The sum of the first three terms inside the squared term on the right-hand side is $O_p(1)$, which yields $r^{1/2}\lambda \tilde{w}_k\operatorname{sgn}(\check{\beta}_k)=O_p(1)$.
Namely, there exists a constant $M'$ such that for sufficiently large $r$,
\[
\operatorname{P}\big(\check{\beta}_k \neq 0,\, |\check{\beta}_k| < M r^{-1/2},\, r^{1/2} \lambda \tilde{w}_k>M'\big)
\leq\frac{\epsilon}{4}.
\]
Moreover, under condition (C7)(ii), $\operatorname{P}\big(\check{\beta}_k \neq 0,\, |\check{\beta}_k| < M r^{-1/2},\, r^{1/2} \lambda \tilde{w}_k\leq M'\big)
\leq\operatorname{P}\big(r^{1/2} \lambda \tilde{w}_k\leq M'\big)\leq\frac{\epsilon}{4}$. Hence,
\[
\operatorname{P}(C_k)\leq\frac{\epsilon}{2}+\frac{\epsilon}{4}+\frac{\epsilon}{4}=\epsilon.
\]
The second part of conclusion (b) is then proved. It follows that
\[
o_p(1) =\boldsymbol{S}_{r1}^\Phi(\boldsymbol{\beta}_0) + \boldsymbol{H}_{n11}(\check{\boldsymbol{\beta}}_1-\boldsymbol{\beta}_{01}) - \lambda \tilde{\boldsymbol{w}}_1\operatorname{sgn}(\check{\boldsymbol{\beta}}_1),
\]
where $\tilde{\boldsymbol{w}}_1$ denotes the first $s$ components of $\tilde{\boldsymbol{w}}$. Similar to the proof in \cite{Johnson08}, we have
\[
\boldsymbol{V}_{c11}^{-1/2}\boldsymbol{H}_{n11}\big(\check{\boldsymbol{\beta}}_1 - \boldsymbol{\beta}_{01} + \boldsymbol{H}_{n11}^{-1}\boldsymbol{b}_n\big) \stackrel{d}{\to} N\big(\boldsymbol{0}, \boldsymbol{I}_s\big).
\]

To verify part (c), consider the parameter vector $\boldsymbol{\beta}_1\in\mathbb{R}^s$ lying on the boundary of the sphere centered at $\boldsymbol{\beta}_{01}$, namely $\boldsymbol{\beta}_1 = \boldsymbol{\beta}_{01} + r^{-1/2}\boldsymbol{d}$,
where $\|\boldsymbol{d}\| = l$ and $l$ is a fixed constant. For the penalized estimating functions, we obtain
\[
\begin{aligned}
&r^{1/2}(\boldsymbol{\beta}_1 - \boldsymbol{\beta}_{01})^\top \boldsymbol{H}_{n11}^\top \boldsymbol{S}_{r1}^P(\boldsymbol{\beta}) \\
=\ &r^{1/2}(\boldsymbol{\beta}_1 - \boldsymbol{\beta}_{01})^\top \boldsymbol{H}_{n11}^\top\big[\boldsymbol{S}_{r1}^\Phi(\boldsymbol{\beta})-\lambda \tilde{\boldsymbol{w}}_1\operatorname{sgn}(\boldsymbol{\beta}_1)\big] \\
=\ &r^{1/2}(\boldsymbol{\beta}_1 - \boldsymbol{\beta}_{01})^\top \boldsymbol{H}_{n11}^\top\big[\boldsymbol{S}_{r1}^\Phi(\boldsymbol{\beta}_0)+ \boldsymbol{H}_{n11}\big(\boldsymbol{\beta}_1 - \boldsymbol{\beta}_{01}\big)
-\lambda \tilde{\boldsymbol{w}}_1\operatorname{sgn}(\boldsymbol{\beta}_1)\big] \\
=\ &O_p\big(\|\boldsymbol{\beta}_1 - \boldsymbol{\beta}_{01}\|\big)+r^{1/2}\big(\boldsymbol{\beta}_1 - \boldsymbol{\beta}_{01}\big)^{\top}\boldsymbol{H}_{n11}^\top\boldsymbol{H}_{n11}\big(\boldsymbol{\beta}_1 - \boldsymbol{\beta}_{01}\big)\\
&\quad -r^{1/2}\big(\boldsymbol{\beta}_1 - \boldsymbol{\beta}_{01}\big)^{\top}\boldsymbol{H}_{n11}^{\top}\lambda \tilde{\boldsymbol{w}}_1\operatorname{sgn}(\boldsymbol{\beta}_1).
\end{aligned}
\]
Since $\boldsymbol{H}_{n11}$ is nonsingular, the second term on the right-hand side is greater than $a_0 l^2r^{-1/2}$, where $a_0$ denotes the minimum eigenvalue of $\boldsymbol{H}_{n11}^\top \boldsymbol{H}_{n11}$.
The order of the first term is $O_p(lr^{-1/2})$.
In the third term, $r^{1/2}\lambda \tilde{\boldsymbol{w}}_1\to 0$, and the third term is dominated by the second term.
Thus, for any $\epsilon>0$, if $l$ is chosen sufficiently large, the probability that the absolute value of the first term exceeds the second term is less than $\epsilon$ for sufficiently large $r$.
It follows that
\[
\operatorname{P}\left[\min_{\|\boldsymbol{\beta}_1-\boldsymbol{\beta}_{01}\|=r^{-1/2}l} r^{1/2}\left(\boldsymbol{\beta}_1 - \boldsymbol{\beta}_{01}\right)^\top \boldsymbol{H}_{n11}^\top \boldsymbol{S}_{r1}^P\left((\boldsymbol{\beta}_1^\top,\boldsymbol{0}^\top)^\top\right)>0\right]>1-\epsilon.
\]
Applying the Brouwer fixed-point theorem to the continuous function $\boldsymbol{S}_{r1}^P\left((\boldsymbol{\beta}_1^\top,\boldsymbol{0}^\top)^\top\right)$,
it is concluded that $\boldsymbol{H}_{n11}^\top \boldsymbol{S}_{r1}^P(\boldsymbol{\beta})$ admits a solution inside the sphere whenever $\min_{\|\boldsymbol{\beta}_1-\boldsymbol{\beta}_{01}\|=r^{-1/2}l} r^{1/2}\left(\boldsymbol{\beta}_1 - \boldsymbol{\beta}_{01}\right)^\top \boldsymbol{H}_{n11}^\top \boldsymbol{S}_{r1}^P(\boldsymbol{\beta})>0$.
Equivalently, $\boldsymbol{S}_{r1}^P\left((\boldsymbol{\beta}_1^\top,\boldsymbol{0}^\top)^\top\right)$ has a solution inside the sphere.
That is, there exists an exact solution $\check{\boldsymbol{\beta}} = \big(\check{\boldsymbol{\beta}}_1^\top, \boldsymbol{0}^\top\big)^\top$ satisfying $\boldsymbol{S}_{r1}^P(\boldsymbol{\beta})=0$ and $\check{\boldsymbol{\beta}} = \boldsymbol{\beta}_0 + O_p(r^{-1/2})$.

\end{proof}

\end{document}